\def\year{2020}\relax
\documentclass[letterpaper]{article} % DO NOT CHANGE THIS
\usepackage{aaai20}  % DO NOT CHANGE THIS
\usepackage{times}  % DO NOT CHANGE THIS
\usepackage{helvet} % DO NOT CHANGE THIS
\usepackage{courier}  % DO NOT CHANGE THIS
\usepackage[hyphens]{url}  % DO NOT CHANGE THIS
\usepackage{graphicx} % DO NOT CHANGE THIS
\usepackage{graphicx}  % DO NOT CHANGE THIS
\title{Path Planning Problems with Side Observations---When Colonels Play Hide-and-Seek}
\author{Dong Quan Vu,\textsuperscript{\rm 1}
{Patrick Loiseau},\textsuperscript{\rm 2}
{Alonso Silva},\textsuperscript{\rm 3}
{Long Tran-Thanh}\textsuperscript{\rm 4}\\
\textsuperscript{\rm 1}Nokia Bell Labs France, AAAID Department, 
\textsuperscript{\rm 2}Univ. Grenoble Alpes,
Inria, CNRS, Grenoble INP, LIG \& MPI-SWS, \\
\textsuperscript{\rm 3}Safran Tech, Signal and Information Technologies, 
 \textsuperscript{\rm 4}Univ. of Southampton,  School of Electronics and Computer Science\\
 quan\_dong.vu@nokia.com, patrick.loiseau@inria.fr, alonso.silva-allende@safrangroup.com, l.tran-thanh@soton.ac.uk
}

\usepackage{booktabs}       % professional-quality tables
\usepackage{nicefrac}       % compact symbols for 1/2, etc.
\usepackage{amsmath,amsfonts,amsthm,amssymb,mathtools}
\usepackage{thmtools,thm-restate}	%for re-state a lemma in Appendix
\DeclarePairedDelimiter{\ceil}{\lceil}{\rceil}
\usepackage[utf8]{inputenc}
\usepackage{graphicx}
\usepackage{epstopdf}
\usepackage{algorithm}
\usepackage{algorithmic}
\allowdisplaybreaks	
\usepackage{subfig}

\newtheorem{theorem}{Theorem}[section]
\newtheorem{proposition}[theorem]{Proposition}
\newtheorem{corollary}[theorem]{Corollary}
\newtheorem{lemma}[theorem]{Lemma}

\newtheorem{definition}[theorem]{Definition}

\newcommand\mO{\mathcal{O}}
\newcommand\Loss{\boldsymbol{\mathcal{\ell}}}
\newcommand\hLoss{\hat{\boldsymbol{\mathcal{\ell}}}}

\newcommand\pa{\boldsymbol{p}}
\newcommand\pset{\mathcal{P}}
\newcommand\patil{\tilde{\boldsymbol{p}}}
\newcommand\vset{\mathcal{V}}
\newcommand\eset{\mathcal{E}}	
\newcommand\ob{\mathbb{O}}

\def\OCO/{\textsc{OComb}}
\def\Algo/{\textsc{Exp3-OE}}

\begin{document}

\maketitle

\begin{abstract}
Resource allocation games such as the famous Colonel Blotto (CB) and Hide-and-Seek (HS) games are often used to model a large variety of practical problems, but only in their one-shot versions. Indeed, due to their extremely large strategy space, it remains an open question how one can efficiently learn in these games. In this work, we show that the online CB and HS games can be cast as path planning problems with side-observations (SOPPP): at each stage, a learner chooses a path on a directed acyclic graph and suffers the sum of losses that are adversarially assigned to the corresponding edges; and she then receives semi-bandit feedback with side-observations (i.e., she observes the losses on the chosen edges plus some others). We propose a novel algorithm, \Algo/, the first-of-its-kind with guaranteed efficient running time for SOPPP without requiring any auxiliary oracle. We provide an expected-regret bound of \Algo/ in SOPPP matching the order of the best benchmark in the literature. Moreover, we introduce additional assumptions on the observability model under which we can further improve the regret bounds of \Algo/. We illustrate the benefit of using \Algo/ in SOPPP by applying it to the online CB and~HS~games.  
\end{abstract}

%%%%%%%%%%%%%%%%%%%%%%%%%%%%%%%%%%%%%%%%%%%%%%%%%%%%%%%%%%%%%%%%%%%%%%%%%%%%%%%%%%%%%%
%%%%%%%%%%%%%%%%%%%%%%%%%%%%%%%%%%%%%%%%%%%
% %
\section{Introduction}
% %
\label{intro}
%Paragraph 1: What is CB and HS, why do you care about them?
Resource allocation games have been studied profoundly in the literature and showed to be very useful to model many practical situations, including online decision problems, see e.g. \cite{blocki2013audit,bower2005resource,korzhyk2010complexity,zhang2009multi}. In particular, two of the most renowned are the Colonel Blotto game (henceforth, CB game) and the Hide-and-Seek game (henceforth, HS game). In the (one-shot) \emph{CB game}, two players, each with a fixed amount of budget, simultaneously allocate their indivisible resources (often referred to as troops) on $n \in \mathbb{N}$ battlefields, each player's payoff is the aggregate of the values of battlefields where she has a higher allocation. The scope of applications of the CB games includes a variety of problems; for instance, in security where resources correspond to security forces (e.g., \cite{chia2012,schwartz2014}), in politics where budget are distributed to attract voters (e.g., \cite{kovenock2012,roberson2006}), and in advertising for distributing the ads' broadcasting time (e.g., \cite{masucci2014,masucci2015}). On the other hand, in the (one-shot) \emph{HS game}, a seeker chooses $n$ among $k$ locations ($n \le k$) to search for a hider, who chooses the probability of hiding in each location. The seeker's payoff is the summation of the probability that the hider hides in the chosen locations and the hider's payoff is the probability that she successfully escapes the seeker's pursuit. Several variants of the HS games are used to model surveillance situations \cite{bhattacharya2014surveillance}, anti-jamming problems \cite{navda2007using,wang16}, vehicles control \cite{vidal2002probabilistic}, etc. 

% This paragraph: Offline sucks, online is more suitable. Why?
Both the CB and the HS games have a long-standing history (originated by~\cite{borel1921} and~\cite{vonneumann53}, respectively); however, the results achieved so-far in these games are mostly limited to their one-shot and full-information version (see e.g., \cite{Behnezhad17a,grosswagner,roberson2006,schwartz2014,Vu18a} for CB games and \cite{hespanha2000probabilistic,yavin1987pursuit} for HS games). On the contrary, in most of the applications (e.g.,  telecommunications, web security, advertising), a more natural setting is to consider the case where the game is played repeatedly and players have access only to incomplete information at each stage. In this setting, players are often required to sequentially learn the game on-the-fly and adjust the trade-off between exploiting known information and exploring to gain new information. Thus, this work focuses on the following sequential learning problems:

\noindent $(i)$ The \emph{online CB game}: fix $k,n \in \mathbb{N}$ ($k,n \ge 1$); at each stage, a learner who has the budget $k$ plays a CB game against some adversaries across $n$ battlefields; at the end of the stage, she receives limited feedback that is the gain (loss) she obtains from each battlefield (but not the adversaries' strategies). The battlefields' values can change over time and they are unknown to the learner before making the decision at each stage. This setting is generic and covers many applications of the CB game. For instance, in radio resource allocation problem (in a cognitive radio network), a solution that balances between efficiency and fairness is to provide the users fictional budgets (the same budget at each stage) and let them bid across $n$ spectrum carriers simultaneously to compete for obtaining as many bandwidth portions as possible, the highest bidder to each carrier wins the corresponding bandwidth (see e.g., \cite{chien2019stochastic}). At the end of each stage, each user observes her own data rate (the gain/loss) achieved via each carrier (corresponding to battlefields' values) but does not know other users' bids. Note that the actual data rate can be noisy and change over time. Moreover, users can enter and leave the system so no stochastic assumption shall be made for the~adversaries'~decisions.

\noindent $(ii)$ The \emph{online HS game}: fix $k,n \in \mathbb{N}$ ($k,n \ge 1$ and $n \le k$); at each stage, the learner is a seeker who plays the same HS game (with $k$ and $n$) against an adversary; at the end of the stage, the seeker only observes the gains/losses she suffers from the locations she chose. This setting is practical and one of the motivational examples is the spectrum sensing problem in opportunistic spectrum access context (see e.g., \cite{yucek2009survey}). At each stage, a secondary user (the learner) chooses to send the sensing signal to at most $n$ among $k$ channels (due to energy constraints, she cannot sense all channels) with the objective of sensing the channels with the availability as high as possible. The leaner can only measure the reliably (the gain/loss) of the channels that she sensed. Note that the channels' availability depend on primary users' decisions that is non-stochastic. 

A formal definition of these problems is given in Section~\ref{sec:games}; hereinafter, we reuse the term CB game and HS game to refer to this sequential learning version of the games. The main challenge here is that the strategy space is exponential in the natural parameters (e.g., number of troops and battlefields in the CB game, number of locations in the HS game); hence how to efficiently learn in these games is an open question. 

Our \textbf{first contribution} towards solving this open question is to show that the CB and HS games can be cast as a \emph{Path Planning Problem} (henceforth, PPP), one of the most well-studied instances of the \emph{Online Combinatorial Optimization} framework (henceforth, \OCO/; see \cite{chen2013combinatorial} for a survey). In PPPs, given a directed graph a source and a destination, at each stage, a learner chooses a path from the source to the destination; simultaneously, a loss is adversarially chosen for each edge; then, the learner suffers the aggregate of edges' losses belonging to her chosen path. The learner's goal is to minimize regret. The information that the learner receives in the CB and HS games as described above straightforwardly corresponds to the so-called \emph{semi-bandit} feedback setting of PPPs, i.e., at the end of each stage, the learner observes the edges' losses belonging to her chosen path (see Section~\ref{sec:games} for more details). However, the specific structure of the considered games also allows the learner to deduce (without any extra cost) from the semi-bandit feedback the losses of some of the other edges that may not belong to the chosen path; these are called \emph{side-observations}. Henceforth, we will use the term SOPPP to refer to this {PPP under semi-bandit feedback with~side-observations}. 

SOPPP is a special case of \OCO/ with side-observations (henceforth, S\OCO/) studied by \cite{kocak14} and, following their approach, we will use \emph{observation graphs}\footnote{The observation graphs, proposed by \cite{kocak14} and used here for SOPPP, extend the side-observations model for multi-armed bandits problems studied by \cite{alon15,alon13,mannor11}. Indeed, they capture side-observations between edges whereas the side-observations model considered by \cite{alon15,alon13,mannor11} is between actions, i.e., paths in PPPs.} (defined in Section~\ref{sec:SOPPPFor}) to capture the learner's observability. 
\cite{kocak14} focuses on the class of Follow-the-Perturbed-Leader (FPL) algorithms (originated from \cite{kalai2005efficient}) and proposes an algorithm named {FPL-IX} for S\OCO/, which could be applied directly to SOPPP. However, this faces two main problems: ($i$) the efficiency of FPL-IX is only guaranteed with high-probability (as it depends on the geometric sampling technique) and it is still super-linear in terms of the time horizon, thus there is still room for improvements; ($ii$) FPL-IX requires that there exists an efficient oracle that solves an optimization problem at each stage. Both of these issues are incompatible with our goal of learning in the CB and HS~games: although the probability that FPL-IX fails to terminate is small, this could lead to issues in implementing it in practice where the
learner is obliged to quickly give a decision in each stage; it is unclear which oracle should be used in applying FPL-IX to the CB and HS~games. 

In this paper, we focus instead on another prominent class of \OCO/ algorithms, called \textsc{Exp3} \cite{auer02b,Freund1997}. One of the key open questions in this field is how to design a variant of \textsc{Exp3} with efficient running time and good regret guarantees for \OCO/ problems in each feedback setting (see, e.g., \cite{cesa2012}). Then, our \textbf{second contribution} is to propose an \textsc{Exp3}-type algorithm for SOPPPs that solves both of the aforementioned issues of FPL-IX and provides good regret guarantees; i.e., we give an affirmative answer to an important subset of the above-mentioned open problem. In more details, this contribution is three-fold: $(i)$ We propose a \emph{novel algorithm, \Algo/,} that is applicable to any instance of SOPPP. Importantly, \Algo/ is always guaranteed to run efficiently (i.e., in polynomial time in terms of the number of edges of the graph in SOPPP) without the need of any auxiliary oracle; $(ii)$ We prove that \Algo/ guarantees an upper-bound on the expected regret matching in order with the best benchmark in the literature (the FPL-IX algorithm). We also prove further improvements under additional assumptions on the observation graphs that have been so-far ignored in the literature; $(iii)$ We demonstrate the benefit of using the \Algo/ algorithm in the CB and HS games.

Note importantly that the SOPPP model (and the \Algo/ algorithm) can be applied into many problems beyond the CB and HS games, e.g., auctions, recommendation systems. To highlight this and for the sake of conciseness, we first study the generic model of SOPPP in Section~\ref{sec:SOPPPFor} and present our second contribution in Section~\ref{sec:Exp3-OE}, i.e., the \Algo/ algorithm in SOPPPs; we delay the formal definition of the CB and HS games, together with the analysis on running \Algo/ in these games (i.e., our first contribution) to Section~\ref{sec:games}. 

Throughout the paper, we use bold symbols to denote vectors, e.g., $\boldsymbol{z} \in \mathbb{R}^n$, and $\boldsymbol{z}(i)$ to denote the $i$-th element. 
%of a vector (e.g.,
%$\boldsymbol{z}\! =\! (\boldsymbol{z}(1), \boldsymbol{z}(2), \ldots, \boldsymbol{z}(n))$).
For any $m\ge 1$, the set $\{1,2, \ldots, m\}$ is denoted by $[m]$ and
% $R^n_{\ge 0}$ denotes the set of all $n$-tuples whose elements are non-negative ($n\ge 2$). 
the indicator function of a set $A$ is denoted by $\mathbb{I}_A$. For graphs, we write either $e\! \in \!\pa$ or $\pa\! \ni \! e$ to refer that an edge $e$ belongs to a path~$\pa$. Finally, we use $\tilde{\mO}$ as a version of the big-O asymptotic notation that ignores the logarithmic~terms.

%%%%%%%%%%%%%%%%%%%%%%%%%%%%%%%%%%%%%%%%%%%%%%%%%%%%%%%%%%%%%%%%%%%%%%%%%%%%%%%%%%%%%
%%%%%%%%%%%%%%%%%%%%%%%%%%%%%%%%%%%%%%%%%%
%
%%%%%%%%%%%%%%%%%%%%%%%%%%%%%%%%%%%%%%%%%%%%%%%%%%%%%%%%%%%%%%%%%%%%%%%%%%%%%%%%%%%%%
%%%%%%%%%%%%%%%%%%%%%%%%%%%%%%%%%%%%%%%%%%
% %
\section{Path Planning Problems with Side-Observations (SOPPP) Formulation}
\label{sec:SOPPPFor}
% %
As discussed in Section~\ref{intro}, motivated by the CB and HS games, we propose the path planning problem with semi-bandit and side-observations feedback (SOPPP).

\textbf{SOPPP model.} Consider a directed acyclic graph (henceforth, DAG), denoted by $G$, whose set of vertices and set of edges are respectively denoted by $\vset$ and $\eset$. Let $V:=|\vset| \ge 2$ and $E:= |\eset| \ge 1$; there are two special vertices, a source and a destination, that are respectively called $s$ and $d$. We denote by $\pset$ the set of all \textit{paths} starting from $s$ and ending at $d$; let us define $P:= |\pset|$. Each path $\pa \in \pset$ corresponds to a vector in $\{0,1\}^{E}$ (thus, $\pset \subset \{0,1\}^E$) where $\pa(e) = 1$ if and only if edge $e\in \eset$ belongs to~$\pa$. Let $n$ be the length of the longest path in $\pset$, that is \mbox{$\| \pa \| _{1} \le n, \forall \pa \in \pset$}. Given a time horizon $T \in \mathbb{N}$, at each (discrete) stage $t \in [T]$, a \emph{learner} chooses a path \mbox{$\patil_t \in \pset$}. Then, a \emph{loss vector} \mbox{$\Loss_t \in [0,1]^E$} is secretly and adversarially chosen. Each element \mbox{$\Loss_t(e)$} corresponds to the scalar loss embedded on the edge $e \in \eset$. Note that we consider the \emph{non-oblivious adversary}, i.e., $\Loss_t$ can be an arbitrary function of the learner's past actions $\patil_s, \forall s \in [t-1]$, but not~$\patil_t$.\footnote{This setting is considered by most of the works in the non-stochastic/adversarial bandits literature, e.g.,~\cite{alon13,cesa2012}.} The learner's incurred loss is \mbox{$L_t(\patil_t)=(\patil_t)^{\top} \Loss_t =  \sum_{e \in \patil_t}\nolimits{\Loss_t(e)}$}, i.e., the sum of the losses from the edges belonging to~$\patil_t$. The learner's feedback at stage $t$ after choosing $\patil_t$ is presented as follows. First, she receives a \textit{semi-bandit} feedback, that is, she observes the edges' losses $\Loss_{t}(e)$, for any $e$ belonging to the chosen path~$\patil_t$. Additionally, each edge \mbox{$e \in \patil_t$} may reveal the losses on several other edges. To represent these \textit{side-observations} at time~$t$, we consider a graph, denoted $G_t^O$, containing $E$ vertices. Each vertex $v_e$ of $G_t^O$ corresponds to an edge $e \in \eset$ of the graph~$G$. There exists a directed edge from a vertex $v_e$ to a vertex $v_{e^{\prime}}$ in $G_t^O$ if, by observing the edge loss $\Loss_{t}(e)$, the learner can also deduce the edge loss~$\Loss_t({e^{\prime}})$; we also denote this by~\mbox{$e \rightarrow {e^{\prime}}$} and say that the edge $e$ reveals the edge $e^{\prime}$. The objective of the learner is to minimize the cumulative \textit{expected regret}, defined as \mbox{$R_{T} := \mathbb{E}\left[ {\sum \nolimits_{t \in [T]} L\left( {\patil}_t\right)} \right]- \min\limits_{{\pa^*} \in \pset} {\sum \nolimits_{t \in [T]}{L}\left( {\pa^*}\right)} $}.

Hereinafter, in places where there is no ambiguity, we use the term \emph{path} to refer to a path in $\pset$ and the term \emph{observation graphs} to refer to~$G_t^O$. In general, these observation graphs can depend on the decisions of both the learner and the adversary. On the other hand, all vertices in $G_t^O$ always have self-loops. In the case where none among $G_t^O, t\in [T]$ contains any other edge than these self-loops, no side-observation is allowed and the problem is reduced to the classical semi-bandit setting. If all $G_t^O,t \in [T]$ are complete graphs, SOPPP corresponds to the full-information PPPs. In this work, we focus on considering the \emph{uninformed setting}, i.e., the learner observes $G_t^O$ only after making a decision at time~$t$. On the other hand, we introduce two new~notations:
% %
\begin{align*}
    &  \ob_t(e) \!:=\!  \left\{\pa \in \pset: \! \exists e^{\prime} \! \in \! \pa, e^{\prime} \! \rightarrow  e \right\}, \forall e \in \! \eset, \\
     & \ob_t(\pa) \!:= \! \left\{e \in\! \eset:\! \exists e^{\prime} \in \! \pa, e^{\prime}\! \rightarrow e \right\}, \forall \pa \in \!\pset. \nonumber 
\end{align*}
Intuitively, $\ob_t(e)$ is the set of all paths that, if chosen, reveal the loss on the edge $e$ and $\ob_t(\pa)$ is the set of all edges whose losses are revealed if the path $\pa$ is chosen. Trivially, $\pa \in \ob(e) \Leftrightarrow e \in \ob(\pa)$. Moreover, due to the semi-bandit feedback, if $\pa^{*} \ni e^{*}$, then $\pa^{*} \in \ob_t(e^{*})$ and $e^{*} \in \ob_t(\pa^{*})$. Apart from the results for general observation graphs, in this work, we additionally present several results under two particular assumptions, satisfied by some instances in practice (e.g., the CB and HS games), that provide more refined regret~bounds compared to cases that were considered by \cite{kocak14}: 
\begin{trivlist}
    \item[$(i)$]\textit{symmetric} observation graphs where for each edge from $v_e$ to $v_{e^{\prime}}$, there also exists an edge from $v_{e^{\prime}}$ to $v_e$ (i.e., if $e \rightarrow e^{\prime}$ then $e^{\prime} \rightarrow e $); i.e., $G^O_t$ is an undirected graph; 
    \item[$(ii)$]observation graphs under the following \emph{assumption~$(A0)$} that requires that if two edges belong to a path in $G$, then they cannot simultaneously reveal the loss of another~edge: 
    \noindent \textbf{Assumption $\boldsymbol{(A0)}$:} \emph{For any $e \!\in\! \eset$, if $e^{\prime} \! \rightarrow \! e$ and $e^{\prime \prime} \! \rightarrow \! e$, then $\nexists \pa \in \pset: \pa \ni e^{\prime}, \pa \ni e^{\prime \prime}$}.
\end{trivlist}
%

%

%
%%%%%%%
%%%%%%%%%%%%%%%%%%%
%%%%%%%%%%%%%%%%%%%%%%%%%
%%%%%%%%%%%%%%%%%%%%%%%%
%%%%%%%%%%%%%%%%%%%%%%%%%%%%%%
%%%%%%%%%%%%%%%%%%%%%%%%%%%%%%%%%%%
%%%%%%%%%%%%%%%%%%%%%%%%%%%%%%%%%%%

%%%%%%%%%%%%%%%%%%%%%%%%%%%%%%%%%%%%%%%%%%%%%%%%%%%%%%%%%%%%%%%%%%%%%%%%%%%%%%%%%%%%%
%%%%%%%%%%%%%%%%%%%%%%%%%%%%%%%%%%%%%%%%%
%
% %
\section{\Algo/ - An Efficient Algorithm for the~SOPPP}
\label{sec:Exp3-OE}
%
% %
In this section, we present a new algorithm for SOPPP, called \Algo/ (OE stands for Observable Edges), whose pseudo-code is given by Algorithm~\ref{OEdgeAl}. The guarantees on the expected regret of \Algo/ in SOPPP is analyzed in Section~\ref{sec:OEPerform}. Moreover, $\Algo/$ always runs efficiently in polynomial time in terms of the number of edges of $G$; this is discussed in Section~\ref{sec:Effi}.

\begin{algorithm}[H]
   \centering
   \captionof{algorithm}{\Algo/~Algorithm for SOPPP.}
   \label{OEdgeAl}
        \begin{algorithmic}[1]
           \STATE {\bfseries Input:} $T$, $\eta, \beta >0$, graph $G$.
           \STATE Initialize $w_1(e):= 1$, $\forall e \in \eset $.
           \FOR{$t=1$ {\bfseries to} $T$}
           \STATE Loss vector $\Loss_t$ is chosen adversarially (unobserved).
           \STATE Use WP Algorithm (see Appendix~\ref{sec:AppenWeiPush}) to sample a path $\patil_t$ according to \mbox{$x_t(\patil_t)$} (defined in \eqref{pathwei}).
            \STATE Suffer the loss $L_t(\patil_t)= \sum \nolimits_{e \in \patil_t} {\Loss_{t}(e)}$.
           \STATE Observation graph~$G_t^O$ is generated and $\Loss_{t}(e)$, \mbox{$\forall e \in \ob_t{(\patil_t)}$} are observed.
           \STATE \mbox{$\hLoss_{t}(e) \!:=\! {\Loss_{t}(e)}\mathbb{I}_{\{e \in \ob_t(\patil_t) \}} \big/ {(q_t(e) + \beta)} $}, \mbox{$\forall e \! \in \! \eset$}, where \mbox{$q_t(e)\!:=\! \sum \nolimits_{\pa \in \ob_t(e)}{x_t(\pa)}$} is computed by~Algorithm~\ref{Algo:EstLoss} (see Section~\ref{sec:Effi}).
           \STATE Update weights $w_{t+1}(e):= w_{t}(e) \cdot \exp(-\eta \hLoss_t(e))$.
           \ENDFOR
     \end{algorithmic}
\end{algorithm}
As an \textsc{Exp3}-type algorithm, \Algo/ relies on the average weights sampling where at stage $t$ we update the weight $w_t(e)$ on each edge $e$ by the exponential rule (line $9$). For each path $\pa$, we denote the path weight \mbox{$w_t(\pa): = \prod \nolimits_{e \in \pa}{w_t(e)}$} and define the following~terms:
\begin{equation}
    x_t(\pa) := \frac{{\prod \limits_{e \in \pa}{w_t(e)}}}  {\sum \limits_{\pa^{\prime} \in \pset} {\prod \limits_{e^{\prime} \in \pa^{\prime}}{w_t(e^{\prime} )} }} 
    =  \frac{{w_t(\pa)}} {\sum \limits_{\pa^{\prime} \in \pset}{w_t(\pa^{\prime})}}, \forall \pa \in \pset. \label{pathwei}
\end{equation}
Line 5 of \Algo/ involves a sub-algorithm, called the WPS algorithm, that samples a path $ \pa\in\pset$ with probability $x_t(\pa)$ (the sampled path is then denoted by $\patil_t$) from any input $\{w_t(e), e\in \eset\}$ at each stage $t$. This algorithm is based on a classical technique called weight pushing (see e.g., \cite{takimoto2003,gyorgy2007}). We discuss further details and present an explicit formulation of the WPS algorithm in~Appendix~\ref{sec:AppenWeiPush}).

Compared to other instances of the \textsc{Exp3}-type algorithms, \Algo/ has two major differences. First, at each stage $t$, the loss of each edge $e$ is estimated by $\hLoss_t(e)$ (line $8$) based on the term $q_t(e)$ and a parameter $\beta$. Intuitively, $q_t(e)$ is the probability that the loss on the edge $e$ is revealed from playing the chosen path at $t$. Second, the implicit exploration parameter $\beta$ added to the denominator allows us to ``pretend to explore" in \Algo/ without knowing the observation graph $G_t^O$ before making the decision at stage $t$ (the uninformed~setting). Unlike the standard \textsc{Exp3}, the loss estimator used in \Algo/ is \emph{biased}, i.e., for any $e \in \eset$,
%
%
% %
\begin{align}
     \mathbb{E}_t\left[ \hLoss_t(e) \right] \!
     & = \!\sum \limits_{\patil \in \pset} {x_t(\patil) {\frac{\Loss_{t}(e)}{q_t(e)\!+\!\!\beta} \mathbb{I}_{\{e \in \ob_t(\patil) \}}}}\!  \nonumber \\
            & =\!  \sum \limits_{\patil \in \ob_t(e)} {x_t(\patil) {\frac{\Loss_{t}(e)}{\sum \limits_{\pa \in \ob_t(e)}{x_t(\pa)}\!+\! \beta}}} \! \le \! \Loss_t(e). \label{optimis}
\end{align}
Here, $\mathbb{E}_t$ denotes the expectation w.r.t. the randomness of choosing a path at stage $t$. Second, unlike standard \textsc{Exp3} algorithms that keep track and update on the weight of each path, the weight pushing technique is applied at line $5$ (via the WPS algorithm) and line $8$ (via Algorithm~\ref{Algo:EstLoss} in Section~\ref{sec:Effi}) where we work with edges weights instead of paths weights (recall that $E \ll P$).

%% %
% In this algorithm, the edges weights $w_t(e), e \in \eset$ are updated after each stage $t$ (at line $9$). On the other hand, by an abuse of notation, we define \mbox{$w_t(\pa): = \prod \nolimits_{e \in \pa}{w_t(e)}$} for any $\pa \in \pset$ and call them the path weights. Then, a path $\patil \in \pset$ will be chosen (at line $5$) with probability $x_t(\patil)$ (defined as follows) without the need to keep track and update the terms $w_t(\patil)$ and $x_t(\patil)$:
% % In this algorithm, the edges weights $w_t(e), e \in \eset$ are updated after each stage $t$ while any path $\patil \in \pset$ might be chosen with the probability 
% \begin{equation}
% x_t(\patil)\!:=\! \frac{\prod \nolimits_{e \in \patil}{w_t(e)}}{\sum \nolimits_{\pa \in \pset} {\prod \nolimits_{e \in \pa}{w_t(e)} }}\!=\! \frac{w_t(\patil)}{\sum \nolimits_{\pa \in \pset}{w_t(\pa)}}. \label{pathwei}
% \end{equation}
% %
% % Here, by an abuse of notation, we define \mbox{$w_t(\pa): = \prod \nolimits_{e \in \pa}{w_t(e)}$} for any $\pa \in \pset$ and call them the path weights.
%
%
%
%%%%%%%%%%%%%%%%%%%%%%%%%%%%%%%%%%%%%%%%%%%%
\subsection{Running Time Efficiency of the \Algo/ Algorithm}
\label{sec:Effi}
%
% \begin{minipage}{0.47 \textwidth}
In the WPS algorithm mentioned above, it is needed to compute the terms \mbox{$H_t ( s,u )\!:=\! \sum_{\pa\in \pset_{s,u}}{\prod_{e \in \pa}{w_t(e)} }$} and \mbox{$H_t ( u,d )\! :=\! \sum_{\pa \in \pset_{u\!,d}}{\prod_{e \in \pa}{w_t(e)} }$} for any vertex $u$ in~$G$. Intuitively, $H_t(u,v)$ is the aggregate weight of all paths from vertex $u$ to vertex~$v$ at stage $t$. These terms can be computed recursively in $\mO(E)$ time based on dynamic programming. This computation is often referred to as weight pushing. Following the literature, we present in~Appendix~\ref{sec:AppenWeiPush} an explicit algorithm that outputs $H_t(s,u), H_t(u,d), \forall u$ from any input \mbox{$\{w_t(e), e\in \eset\}$}, called the WP algorithm. Then, a path in $G$ is sampled sequentially edge-by-edge based on these terms by the WPS algorithm. Importantly, the WP and WPS algorithms run efficiently in $\mO(E)$ time.

The final non-trivial step to efficiently implement \Algo/ is to compute $q_t(e)$ in line~$8$, i.e., the probability that an edge $e$ is revealed at stage $t$. Note that $q_t(e)$ is the sum of $|\ob_t(e)| = \mO(P)$ terms; therefore, a direct computation is inefficient while a naive application of the weight pushing technique can easily lead to errors. To compute $q_t(e)$, we propose Algorithm~\ref{Algo:EstLoss}, a non-straightforward application of weight pushing, in which we consecutively consider all the edges \mbox{$e^{\prime} \in \mathfrak{R}_t(e)\!:=\! \{e^{\prime} \in \! \eset \!:\! e^{\prime} \!\rightarrow e \}$}. Then, we take the sum of the terms $x_t(\pa)$ of the paths $\pa$ going through $e^{\prime}$ by the weight pushing technique while making sure that each of these terms $x_t(\pa)$ is included only once, even if $\pa$ has more than one edge revealing~$e$ (this is a non-trivial step). In Algorithm~\ref{Algo:EstLoss}, we denote by $C(u)$ the set of the direct successors of any vertex $u\in \mathcal{V}$. We give a proof that Algorithm~\ref{Algo:EstLoss} outputs exactly $q_t(e)$ as defined in line $8$ of Algorithm~\ref{OEdgeAl} in Appendix~\ref{subalgo}. Algorithm~\ref{Algo:EstLoss} runs in $\mO \left({| \mathfrak{R}_t(e)| E } \right)$ time; therefore, line $8$ of Algorithm~\ref{OEdgeAl} can be done in at most $\mO\left( E ^3 \right)$~time. 
\begin{algorithm}[tb!]
       \captionof{algorithm}{Compute $q_t(e)$ of an edge $e$ at stage $t$.}
       \label{Algo:EstLoss}
    \begin{algorithmic}[1]
    \STATE {\bfseries Input:} $e \in \ob_t(\patil_t)$, set $ \mathfrak{R}_t(e)$ and \mbox{$w_t(\bar{e}), \forall \bar{e} \in \eset$}.
    \STATE Initialize $\bar{w}(\bar{e}):= w_t(\bar{e}), \forall \bar{e} \in \eset$ and $q_t(e):=0$.
    \STATE Compute $H^*(s,d)$ by WP Algorithm (see Appendix~\ref{sec:AppenWeiPush}) with input \mbox{$\{w_t(\bar{e}), \bar{e} \in \eset\}$}. 
    \FOR{$e^{\prime} \in  \mathfrak{R}_t(e)$}
    	\STATE Compute ${H}(s,u), {H}(u,d)$, $\forall u \in \vset$ by WP Algorithm with input \mbox{$\{\bar{w}(\bar{e}), \forall \bar{e} \in \eset\}$}.
        \STATE \mbox{$K(e^{\prime}):= {H}(s,u_{e^{\prime}})\! \cdot \! {w}(e^{\prime}) \! \cdot \! {H}(v_{e^{\prime}},d)$} where edge $e^{\prime}$ goes from $u_{e^{\prime}}$ to $ v_{e^{\prime}} \in C(u_{e^{\prime}})$.
        \STATE $q_t(e):= q_t(e) + K(e^{\prime})/ H^*(s,d)$.
        \STATE Update $\bar{w}(e^{\prime}) = 0$.
        \ENDFOR
        \STATE {\bfseries Output:} $q_t(e)$.
     \end{algorithmic}
\end{algorithm}

In conclusion, $\Algo/$ runs in at most $\mO(E^{3} T)$ time, this guarantee works even for the worst-case scenario. For comparison, the FPL-IX algorithm runs in $\mO(E |\vset|^2T)$ time in expectation and in \mbox{$\tilde{\mO}(n^{1/2}E^{3/2}\ln(E/\delta) T^{3/2})$} time with a probability at least $1-\delta$ for an arbitrary $\delta>0$.\footnote{If one runs FPL-IX with Dijkstra's algorithm as the optimization oracle and with parameters chosen by \cite{kocak14}}
That is, FPL-IX might fail to terminate with a strictly positive probability\footnote{A stopping criterion for FPL-IX can be chosen to avoid this issue but it raises the question on how one chooses the criterion such that the regret guarantees hold.} and it is not guaranteed to have efficient running time in all cases. Moreover, although this complexity bound of FPL-IX is slightly better in terms of $E$, the complexity bound of \Algo/ improves that by a factor of $\sqrt{T}$. As is often the case in no-regret analysis, we consider the setting where T is significantly larger than other parameters of the problems; this is also consistent with the motivational applications of the CB and HS games presented in Section~\ref{intro}. Therefore, our contribution in improving the algorithm’s running time in terms of ${T}$ is relevant. 
%
%
 
%%%%%%%%%%%%%%%%%%%%%%%%%%%%%
%%%%%%%%%%%%%%%%%%%%%%%
%%%%%%%%%%%%%%%%%%%%%%%%%%%%%%%%%%%%
\subsection{Performance of the \Algo/ Algorithm}
\label{sec:OEPerform}
In this section, we present an upper-bound of the expected regret achieved by the \Algo/ algorithm in the SOPPP. For the sake of brevity, with $x_t(\pa)$ defined in \eqref{pathwei}, for any \mbox{$t \in [T]$} and \mbox{$e\in \eset$}, we denote:
\begin{equation*}
    r_t(e):= \sum \nolimits_{\pa \ni e}{x_t(\pa)} \textrm{ and } Q_t:= \sum \nolimits_{e \in \eset} {{r_t(e)} \big/ {(q_t(e)\! + \!\beta)}}.
\end{equation*}
Intuitively, $r_t(e)$ is the probability that the chosen path at stage $t$ contains an edge $e$ and $Q_t$ is the summation over all the edges of the ratio of this quantity and the probability that the loss of an edge is revealed (plus $\beta$). We can bound the expected regret with this key term $Q_t$.
\begin{restatable}{theorem}{Theoone}
\label{maintheo}
    The expected regret of the \Algo/ algorithm in the \emph{SOPPP} satisfies:
    \begin{equation}
        R_T \le {\ln(P)} \big/ {\eta} + \left[\beta + (n \cdot {\eta})\big/ {2}  \right] \cdot \sum \nolimits_{t\in[T]}{ Q_t}. \label{main}
    \end{equation}
\end{restatable}
A complete proof of Theorem~\ref{maintheo} can be found in Appendix~\ref{AppenTheo1} and has an approach similar to \cite{alon13,cesa2012} with several necessary adjustments to handle the new biased loss estimator in \Algo/. To see the relationship between the structure of the side-observations of the learner and the bound of the expected regret, we look for the upper-bounds of $Q_t$ in terms of the observation graphs' parameters. Let $\alpha_t$ be the independence number\footnote{The independence number of a directed graph is computed while ignoring the direction of the edges.} of~$G^O_t$, we have the following statement.
\begin{restatable}{theorem}{theoremQ}
        \label{theoremQ}
        Let us define \mbox{$M:=  \lceil  2 E  ^2/\beta \rceil$}, \mbox{$N_t\!:=\! \ln\! \left(1 \!+ \! \frac{M \! + \! E }{\alpha_t}\!\right)$} and \mbox{$K_t \! := \! \ln \!\left(1 \!+\! \frac{nM \! +\!  E }{\alpha_t}\!\right)$}. Upper-bounds of $Q_t$ in different cases of~$G^O_t$ are given in the following table:
            \begin{table}[H]
                \begin{center}
                    \begin{small}
                    \begin{sc}
                    \begin{tabular}{l c c}
                    \toprule
                    {} & satisfies $(A0)$ & not satisfies $(A0)$  \\
                    \midrule
                    Symmetric  & $\alpha_t$ & $n\alpha_t $ 
                    \\ 
                    Non-Symmetric & $1\!\!+\!2\alpha_{t}N_t$ & $2n\left(1\!\!+\!{\alpha_{t}} K_t \right)$ \\
                    \bottomrule
                    \end{tabular}
                    \end{sc}
                    \end{small}
                \end{center}
           \end{table}
\end{restatable}
A proof of this theorem is given in Appendix~\ref{sec:AppenTheo2}. The main idea of this proof is based on several graph theoretical lemmas that are extracted from \cite{alon13,kocak14,mannor11}. These lemmas establish the relationship between the independence number of a graph and the ratios of the weights on the graph's vertices that have similar forms to the key-term $Q_t$. The case where observation graphs are non-symmetric and do not satisfy assumption~$(A0)$ is the most general setting. Moreover, as showed in Theorem~\ref{theoremQ}, the bounds of $Q_t$ are improved if the observation graphs satisfy either the symmetry condition or assumption~$(A0)$. Intuitively, given the same independence numbers, a symmetric observation graph gives the learner more information than a non-symmetric one; thus, it yields a better bound on $Q_t$ and the expected regret. On the other hand, assumption~$(A0)$ is a technical assumption that allows the use of different techniques in the proofs to obtain better bounds. These cases have not been explicitly analyzed in the literature while they are satisfied by several practical situations, including the CB and HS games (see Section~\ref{sec:games}).

Finally, we give results on the upper-bounds of the expected regret, obtained by the \Algo/ algorithm, presented as a corollary of Theorems~\ref{maintheo} and~\ref{theoremQ}.
\begin{restatable}{corollary}{coroll}
\label{corrol1}
In \emph{SOPPP}, let $\alpha$ be an upper bound of $\alpha_{t}, \forall t\in[T]$. With appropriate choices of the parameters $\eta$ and $\beta$, the expected regret of the {\Algo/}~algorithm is:
    \begin{trivlist}
        \item[$(i)$] \mbox{$R_T \le 
         \tilde{\mO} (n\sqrt{T\alpha \ln(P)}  )$} in the general cases.
        \item[$(ii)$]\mbox{$R_T \le 
         \tilde{\mO}  (\sqrt{ n T\alpha \ln(P}  )$} if assumption~$(A0)$ is satisfied by the observation graphs $G^O_t, \forall t\in[T]$.
     \end{trivlist}
\end{restatable}
A proof of Corollary~\ref{corrol1} and the choices of the parameters~$\beta$ and $\eta$ (these choices are non-trivial) yielding these results will be given in  Appendix~\ref{AppenTheo3}. We can extract from this proof several more explicit results as follows: in the general case, \mbox{$R_T  \le  \mO \left(n \sqrt{T \alpha \ln(P) [1  +  \ln(\alpha + \alpha \ln(\alpha) + E)]} \right)$} when the observations graphs are non-symmetric and \mbox{$R_T\! \le\! (3/2) n\sqrt{T \alpha  \ln(P)} \! +\! \sqrt{n T \alpha} $} if they are all symmetric; on the other hand, in cases that all the observation graphs satisfy $(A0)$, \mbox{$R_T\! \le\! \mO\left(\sqrt{n T\alpha \ln(P) [1 \!+\! 2 \ln(1\!+\!E)]} \right)$} if the observations graphs are non-symmetric and \mbox{$R_T\! \le \! 2\sqrt{n T \alpha  \ln(P)} \!+\! \sqrt{T\alpha}$} if they are all symmetric.

We note that a trivial upper-bound of $\alpha_t$ is the number of vertices of the graph $G^O_t$ which is $E$ (the number of edges in~$G$). In general, the more connected $G^O_t$ is, the smaller $\alpha$ may be chosen; and thus the better upper-bound of the expected regret. In the (classical) semi-bandit setting, $\alpha_t\!=\! E, \forall t\in[T]$ and in the full-information setting, $\alpha_t\! =\! 1$, $\forall t\in[T]$. Finally, we also note that, if $P = \mO(\exp(n))$ (this is typical in practice, including the CB and HS games), the bound in Corollary~\ref{corrol1}-$(i)$ matches in order with the bounds (ignoring the logarithmic factors) given by the \textsc{FPL-IX} algorithm (see \cite{kocak14}). On the other hand, the form of the regret bound provided by the \textsc{Exp3-IX} algorithm (see \cite{kocak14}) does not allow us to compare directly with the bound of \Algo/ in the general SOPPP. \textsc{Exp3-IX} is only analyzed by \cite{kocak14} when $n=1$, i.e., $P=E$; in this case, we observe that the bound given by our \Algo/ algorithm is better than that of \textsc{Exp3-IX} (by some multiplicative~constants).
%%%%%%%%%%%%%%%%%%%%%%%%%%%%%%%%%%%%%%%%%%%%%%%%%%%%%%%%%%%%%%%%%%%%%%%%%%%%%%%%%%%%%
%%%%%%%%%%%%%%%%%%%%%%%%%%%%%%%%%%%%%%%%%%
% %
\section{Colonel Blotto Games and Hide-and-Seek Games as SOPPP}
\label{sec:games}
% In this section, we study a sequential learning framework of the Colonel Blotto game and the Hide-and-Seek game. Although at first glance they seem beyond the scope of the stated problem, these games can be cast as SOPPP problems. These are famous and important problems in game theory with a vast range of applications (see Section~\ref{intro}). In Section~\ref{sec:games}, we present the SOPPP models of these games. In Section~\ref{implement}, we demonstrate the benefit of using the \Algo/ algorithm for learning in these games.
Given the regret analysis of \Algo/ in SOPPP, we now return to our main motivation, the Colonel Blotto and the Hide-and-Seek games, and discuss how to apply our findings to these games. To address this, we define formally the online version of the games and show how these problems can be formulated as SOPPP in Sections~\ref{sec:Blotto} and \ref{sec:Hide}, then we demonstrate the benefit of using the \Algo/ algorithm for learning in these games (Section~\ref{implement}).
%%%%%%%%%%%%%%%%%%
%%%%%%%%%%%%%%%%
%
% %
\begin{figure*}[htb!]%
\centering
    \subfloat[The graph $G_{3,3}$ corresponding to the CB game with \mbox{$k\!=\!n\!=\!3$}. E.g., the bold-blue path represents the strategy~$(0,0,3)$ while the dash-red path represents the strategy~$(2,0,1)$.]{
            \begin{minipage}{4.2cm}
            \hspace*{0.5cm} 
            \includegraphics[height=0.15\textheight]{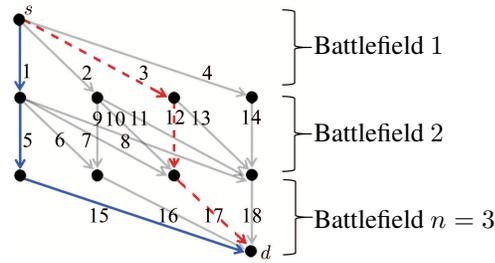}
            \end{minipage}
                \hspace*{0.2cm}
            \begin{minipage}{3.4 cm}
                    Battlefield $1$ \\
            %%\vspace*{0.1cm}
                    \\ \\ 
                    Battlefield $2$\\
                    \\  \\
                    Battlefield $n=3$
            \end{minipage}
        } 
            \hspace*{0.5cm}
        \subfloat[The graph $G_{3,3,1}$ corresponding to the HS game with \mbox{$k\!=\!n\!=\!3$} and $\kappa\! =\!1$. E.g., the blue-bold path represents the $(1,1,1)$ search and the red-dashed path represents the $(2,3,2)$~search.]{
            \begin{minipage}{4.7cm}
            \hspace*{0.8 cm}
            \includegraphics[height=0.15\textheight]{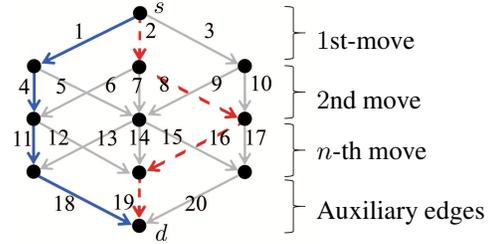}
            \end{minipage}
                \hspace*{0.1cm}    
            \begin{minipage}{3.6 cm}
            \vspace*{0.3cm}
            $1$st-move 
            
            \vspace*{0.4cm}
            $2$nd move
            
            \vspace*{0.3cm}
            $n$-th move
            
            \vspace*{0.4cm}
            Auxiliary edges
            \end{minipage}
}
    \caption{Examples of the graphs corresponding to the CB game and the HS game.}
    \label{fig1}
\end{figure*}
\subsection{Colonel Blotto Games as an SOPPP}
\label{sec:Blotto}
\textbf{The online Colonel Blotto game} (the CB game). This is a game between a learner and an adversary over \mbox{$n\ge 1$}~battlefields within a time horizon~$T>0$. Each battlefield $i~\in~[n]$ has a value $\boldsymbol{b}_t(i)>0$ (unknown to the learner)\footnote{Knowledge on the battlefields' values is not assumed lest it limits the scope of application of our model (e.g., they are unknown in the radio resource allocation problem discussed in Section~\ref{intro}).}
at stage $t$ such that \mbox{$\sum_{i=1}^n{\boldsymbol{b}_t(i)} = 1$}. At stage~$t$, the learner needs to distribute $k$ troops ($k\ge 1$ is fixed) towards the battlefields while the adversary simultaneously allocate hers; that is, the learner chooses a vector $\boldsymbol{z}_t$ in the strategy set \mbox{$S_{k,n}:=\{ \boldsymbol{z} \in \mathbb{N}^n: \sum_{i=1}^n \nolimits{\boldsymbol{z}(i)} = k\}$}. At stage $t$ and battlefield $i\in[n]$, if the adversary's allocation is strictly larger than the learner's allocation $\boldsymbol{z}_t(i)$, the learner loses this battlefield and she suffers the loss $\boldsymbol{b}_t(i)$; if they have tie allocations, she suffers the loss $\boldsymbol{b}_t(i)/2$; otherwise, she wins and suffers no loss. At the end of stage $t$, the learner observes the loss from each battlefield (and which battlefield she wins, ties, or loses) but not the adversary's allocations. The learner's loss at each time is the sum of the losses from all the battlefields. The objective of the learner is to minimize her expected regret. Note that similar to SOPPP, we also consider the non-oblivious adversaries in the CB~game.

While this problem can be formulated as a standard \OCO/, it is difficult to derive an efficient learning algorithm under that formulation, due to the learner's exponentially large set of strategies that she can choose from per stage. Instead, we show that by reformulating the problem as an SOPPP, we will be able to exploit the advantages of the
\Algo/ algorithm to solve it.
To do so, first note that the learner can deduce several side-observations as follows: $(i)$ if she allocates $\boldsymbol{z}_t(i)$ troops to battlefield $i$ and wins, she knows that if she had allocated more than $\boldsymbol{z}_t(i)$ troops to $i$, she would also have won; $(ii)$ if she knows the allocations are tie at battlefield $i$, she knows exactly the adversary's allocation to this battlefield and deduce all the losses she might have suffered if she had allocated differently to battlefield~$i$; $(iii)$ if she allocates $\boldsymbol{z}_t(i)$ troops to battlefield $i$ and loses, she knows that if she had allocated less than $\boldsymbol{z}_t(i)$ to battlefield~$i$, she would also have lost. 
%The learner's objective is to minimize her expected regret.
%% %

%% %
Now, to cast the CB game as SOPPP, for each instance of the parameters $k$ and $n$, we create a DAG $G:=G_{k,n}$ such that the strategy set $S_{k,n}$ has a one-to-one correspondence to the paths set $\pset$ of $G_{k,n}$. Due to the lack of space, we only present here an example illustrating the graph of an instance of the CB game in Figure~\ref{fig1}-(a) and we give the formal definition of $G_{k,n}$~in Appendix~\ref{graph}. The graph $G_{k,n}$ has \mbox{$E\!= \mO(k^2n)$} edges and \mbox{$P=|S_{k,n}|= \Omega\left(2^{\min\{n-1,k\}} \right)$} paths
% \footnote{$E=\! (k\!+\!1)\left[4\! +\! (n\!-\!2)(\!k\!+\!2) \right]\!/2 $ and $P =  \binom{n+k-1}{n-1}$.}
while the length of every path is~$n$. Each edge in $G_{k,n}$ corresponds to allocating a certain amount of troops to a battlefield. Therefore, the CB game model is equivalent to a PPP where at each stage the learner chooses a path in $G_{k,n}$ and the loss on each edge is generated from the allocations of the adversary and the learner (corresponding to that edge) according to the rules of the game.
%
% %% %
At stage $t$, the (semi-bandit) feedback and the side-observations\footnote{E.g., in Figure~\ref{fig1}-(a), if the learner chooses a path going through edge $10$ (corresponding to allocating $1$ troop to battlefield $2$) and wins (thus, the loss at edge $10$ is $0$), then she deduces that the losses on the edges $6,7,8,10,11$, and $13$ (corresponding to allocating at least $1$ troop to battlefield $2$) are all~$0$.} deduced by the learner as described above infers an observation graph $G^O_t$. This formulation transforms any CB game into an~SOPPP.

Note that since there are edges in $G_{m,n} $ that refer to the same allocation (e.g., the edges $5,9,12$, and $14$ in $G_{3,3}$ all refer to allocating~$0$ troops to battlefield~$2$), in the observation graphs, the vertices corresponding to these edges are always connected. Therefore, an upper bound of the independence number $\alpha_t$ of $G^O_t$ in the CB game is \mbox{$\alpha_{\textrm{CB}} = n(k+1) = \mO(nk)$}.
%We note that the side-observations are battlefield-wise; that is that the edges corresponding to an allocation of a battlefield cannot reveal the edges corresponding to another battlefield. 
Moreover, we can verify that the observation graph $G^O_t$ of the CB game \emph{satisfies assumption~$(A0)$} for any $t$ and it is \emph{non-symmetric}.
%
%
%
% %
%%%%%%%%%%%%%%%%%%%%%%%%%%%%%%
%%%%%%%%%%%%%%%%%%%%%%
% %
\subsection{Hide-and-Seek Games as an SOPPP}
\label{sec:Hide}
\textbf{The online Hide-and-Seek game} (the HS game). This is a repeated game (within the time horizon $T>0$) between a hider and a seeker. In this work, we consider that the learner plays the role of the seeker and the hider is the adversary. There are $k$ locations, indexed from $1$ to~$k$. At stage~$t$, the learner sequentially chooses $n$ locations ($1 \le n \le k$), called an {$n$-search}, to seek for the hider, that is, she chooses \mbox{$\boldsymbol{z}_t \in [k]^{n}$} (if $\boldsymbol{z}_t(i)\! =\! j$, we say that location $j$ is her $i$-th move). The hider maliciously assigns losses on all $k$ locations (intuitively, these losses can be the wasted time supervising a mismatch location or the probability that the hider does not hide there, etc.). In the HS game, the adversary is non-oblivious; moreover, in this work, we consider the following condition on how the hider/adversary assigns the losses on the locations:
\begin{trivlist}
\item[$(C1)$] \emph{At stage $t$, the adversary secretly assigns a loss $\boldsymbol{b}_t(j)$ to each location $j \in [k]$ (unknown to the learner). These losses are fixed throughout the $n$-search of the learner. }
\end{trivlist}
The learner's loss at stage $t$ is the sum of the losses from her chosen locations in the $n$-search at stage $t$, that is \mbox{$\sum \nolimits_{i \in [n],j \in [k]}{ {\mathbb{I}_{\{\boldsymbol{z}_t(i) = j \}}} \boldsymbol{b}_t(j) }$}. Moreover, often in practice the $n$-search of the learner needs to satisfy some constraints. In this work, as an example, we use the following constraint: \mbox{$|\boldsymbol{z}_t(i) - \boldsymbol{z}_t(i+1)| \le \kappa, \forall i\in [n]$} for a fixed $\kappa \in [0,k-1]$ (called the \emph{coherence constraint}), i.e., the seeker cannot search too far away from her previously chosen location.\footnote{Our results can be applied to HS games with other constraints, such as \mbox{$\boldsymbol{z}_t(i)  \le \boldsymbol{z}_t(i+1) , \forall i \in [n]$}, i.e., she can only search forward; or, \mbox{$\sum \nolimits_{i \in [n]}{  \mathbb{I}_{\{\boldsymbol{z}_t(i)  = k^* \}}} \le \kappa$}, i.e., she cannot search a location $k^* \in [k]$ more than $\kappa$ times, etc.} At the end of stage $t$, the learner only observes the losses from the locations she chose in her $n$-search, and her objective is to minimize her expected regret over $T$. 

Similar to the case of the CB game, tackling the HS game as a standard \OCO/~is computationally involved. As such, we follow the SOPPP formulation instead. To do this, we create a DAG \mbox{$G:=G_{k,n,\kappa}$} whose paths set has a one-to-one correspondence to the set containing all feasible $n$-search of the learner in the HS game with $k$ locations under $\kappa$-coherent constraint.  Figure~\ref{fig1}-(b) illustrates the corresponding graph of an instance of the HS game and we give a formal definition of $G_{k,n,\kappa}$ in Appendix~\ref{graph}. The HS game is equivalent to the PPP where the learner chooses a path in $G_{k,n,\kappa}$ and edges' losses are generated by the adversary at each stage (note that to ensure all paths end at $d$, there are $n$ auxiliary edges in $G_{k,n,\kappa}$ that are always embedded with $0$ losses). Note that there are \mbox{$E= \mO(k^2 n)$} edges and $P = \Omega(\kappa ^{n-1})$ paths in~$G_{k,n,\kappa}$. Moreover, knowing that the adversary follows condition~$(C1)$, the learner can deduce the following side-observations: within a stage, the loss at each location remains the same no matter when it is chosen among the $n$-search, i.e., knowing the loss of choosing location $j$ as her $i$-th move, the learner knows all the loss if she chooses location $j$ as her $i^{\prime}$-th move for any $i^{\prime} \neq i$. 
The semi-bandit feedback and side-observations as described above generate the observation graphs $G^O_t$ (e.g., in Figure~\ref{fig1}-(b), the edges $1,4,6,11$, and $13$ represent that location $1$ is chosen; thus, they mutually reveal each other). The independence number of $G^O_t$ is $\alpha_{\textrm{HS}} = k$ for any~$t$. The observation graphs of the HS game are \emph{symmetric} and \emph{do not satisfy}~$(A0)$. 
% Finally, we note that by replacing $(C1)$ with other conditions (see, e.g., in \cite{vu2019ArXiV}), we can create other instances of HS games in which the observation graphs are non-symmetric (but they keep the same independence~number).
Finally, we consider a relaxation of condition $(C1)$:
% %
\begin{trivlist}
\item[$(C2)$] \emph{At stage $t$, the adversary assigns a loss $\boldsymbol{b}_t(j)$ on each location $j \in [k]$. For $i = 2, \ldots,n$, after the learner chooses, say location $j_i$, as her $i$-th move, the adversary can observe that and change the losses $\boldsymbol{b}_t(j)$ for any location that has not been searched before by the learner,\footnote{An interpretation is that by searching a location, the learner/seeker ``discovers and secures" that location; therefore, the adversary/hider cannot change her assigned loss at that place.} i.e., she can change the losses $\boldsymbol{b}_t(j), \forall j \notin \{j_1,\ldots, j_i\}$.}
\end{trivlist}
%
% %
By replacing condition $(C1)$ with condition $(C2)$, we can limit the side-observations of the learner: she can only deduce that if $i_1 < i_2$, the edges in $G_{k,n,\kappa}$ representing choosing a location as the $i_1\textrm{-th}$ move reveals the edges representing choosing that same location as the $i_2$-th move; but \textit{not vice versa}. In this case, the observation graph $G^O_t$ is non-symmetric; however, its independence number is still $\alpha_{\rm HS}=k$ as in the HS games with condition~$(C1)$. 
%%%%%%%%
%%%%%%%%%%
% %
\subsection{Performance of \Algo/ in the Colonel Blotto and Hide-and-Seek Games}
\label{implement}
Having formulated the CB game and the HS game as SOPPPs, we can use the \Algo/ algorithm in these games. From Section~\ref{sec:Effi} and the specific graphs of the CB and HS game, we can deduce that \Algo/ runs in at most $\mO(k^6n^3 T)$ time. We remark again that \Algo/'s running time is linear in $T$ and efficient in all cases unlike when we run FPL-IX in the CB and HS games. Moreover, we can deduce the following result directly from Corollary~\ref{corrol1}:
\begin{corollary}
\label{corol2}
The expected regret of the \Algo/ algorithm satisfies:
\begin{trivlist}
\item[$(i)$] $R_T \le \tilde{\mO}(\sqrt{nT \alpha_{\textrm{CB}} \ln(P)}) = \tilde{\mO}(\sqrt{Tn^{3}k})$ in the CB games with $k$ troops and $n$ battlefields.
\item[$(ii)$] $R_T \le \tilde{\mO}(n\sqrt{T \alpha_{\textrm{HS}}\ln(P) }) =  \tilde{\mO}( \sqrt{Tn^3k})$ in the HS games with $k$ locations and $n$-search.
\end{trivlist}
\end{corollary}
%\%

At a high-level, given the same scale on their inputs, the independence numbers of the observation graphs in HS games are smaller than in CB games (by a multiplicative factor of~$n$). However, since assumption~$(A0)$ is satisfied by the observation graphs of the CB games and not by the HS games, the expected regret bounds of the \Algo/ algorithm in these games have the same order of magnitude. From Corollary~\ref{corol2}, we note that in the CB games, the order of the regret bounds given by \Algo/ is better than that of the FPL-IX algorithm (thanks to the fact that $(A0)$ is satisfied).\footnote{More explicitly, in the CB game, FPL-IX has a regret at most $\mO \left(\ln(k^2 n^2 T ) \sqrt{\ln(k^2n)(k^2n^4 \!+\! C n^4 k T)}\right)\!=\!\tilde{\mO}(\sqrt{Tn^4k})$ (C is a constant indicated by \cite{kocak14}) and \Algo/'s regret bound is $ \mO \left(\sqrt{n^2k T \! \cdot\! \min\{n\!-\!1,k\}[1\!+\! 2\ln(1 \!+\! k^2n)] }\right)$ (if \mbox{$n-1 \le k$}, we can rewritten this bound as $\tilde{\mO}(\sqrt{T n^3k})$).} On the other hand, in the HS games with $(C1)$, the regret bounds of the \Algo/ algorithm improves the bound of FPL-IX but they are still in the same order of the games' parameters (ignoring the logarithmic factors).\footnote{More explicitly, in HS games with~$(C1)$, FPL-IX's regret is \mbox{$\mO \left(\ln(k^2 n^2 T) \sqrt{\ln(k^2n)(k^2n^4 + C n^3 k T ) }\right) \!=\!\tilde{\mO}(Tn^3k)$} and $\Algo/$'s regret is \mbox{$\mO\left((3/2)\sqrt{n^3kT\ln(k)} \!+ \! \sqrt{nkT} \right) \!=\! \tilde{\mO}(Tn^3k)$} (similar results can be obtained for the HS games with $(C2)$).} Note that the the regret bound of \Algo/ in the HS game with Condition~$(C1)$ (involving symmetric observation graphs) is slightly better than that in the HS game with Condition~$(C2)$.

We also conducted several numerical experiments that compares the running time and the actual expected regret of \Algo/ and FPL-IX in CB and HS games. The numerical results are in consistent with theoretical results in this work. Our code for these experiments can be found at \url{https://github.com/dongquan11/CB-HS.SOPPP}.

Finally, we compare the regret guarantees given by our \Algo/ algorithm and by the OSMD algorithm (see \cite{AudibertBL2014})---the benchmark algorithm for \OCO/ with semi-bandit feedback (although OSMD does not run efficiently in general): \Algo/ is better than OSMD in CB games if $\mO\left(n \cdot \ln{( n^3 k^5 \sqrt{T})} \right) \le k$; in HS games $(C1)$ if $\mO (n \ln{\kappa} ) \le k$ and in the HS games with condition $(C2)$ if \mbox{$n \cdot \ln{\kappa}  \ln{(n^4 k^5 \sqrt{T})} \le \mO(k)$}.
% ; and in the HS games with condition $(C2)$ if \mbox{$n \cdot \ln{\kappa}  \ln{(n^4 k^5 \sqrt{T})} \le \mO(k)$}
We give a proof of this statement in Appendix~\ref{prooflogarithm}. 
Intuitively, the regret guarantees of \Algo/ is better than that of OSMD in the CB games where the learner’s budget is sufficiently larger than the number of battlefields and in the HS games where the total number of locations is sufficiently larger than the number of moves that the learner can make in each stage.
%

%%%%%%%%%%%%%%%%%%%%%%%%%%%%%%%%%%%%%%%%%%%%%%%%%%%%%
%%%%%%%%%%%%%%%%%%%%%%%%%%%%%%%%%%%%%%%%%%%%%%%%%%%%%%%%
\section{Conclusion}
In this work, we introduce the \Algo/ algorithm for the path planning problem with semi-bandit feedback and side-observations. \Algo/ is always efficiently implementable. Moreover, it matches the regret guarantees compared to that of the \textsc{FPL-IX} algorithm (\Algo/ is better in some cases). We apply our findings to derive the first solutions to the online version of the Colonel Blotto and Hide-and-Seek games.
This work also extends the scope of application of the PPP model in practice, even for large~instances.

\paragraph{Acknowledgment:} This work was supported by ANR through the “Investissements d’avenir” program (ANR-15-IDEX-02) and grant ANR-16-TERC0012; and by the Alexander von Humboldt~Foundation. Partial of this work was done when the authors were at LINCS.

%%%%%%%%%%%%%%%%%%%%%%%%%%%%%%%%%%%%%%%%%%%%%%%%%
{
\fontsize{9.0pt}{10.0pt} \selectfont
\bibliography{mybibfile}
\bibliographystyle{aaai}
}
%%%%%%%%%%%%%%%%%%%%%%%%%%%%%%%%%%%%%%%%%%%%%%%%%%%%%%%%%%%%%%%%%%%%%%%%%%%%%%%
%%%%%%%%%%%%%%%%%%%%%%%%%%%%%%%%%%%%%%%%%%%%%%%%%%%%%%%%%%%%%%%%%%%%%%%%%%%%%%%
%%%%%%%%%%%%%%%%%%%%%%%%%%%%%%%%%%%%%%%%%%%%%%%%%%%%%%%%%%%%%%%%%%%%%%%%%5
%%%%%%%%%%%%%%%%%%%%%%%%%%%%%%%%%%%%%%%%%%%%%%%%%%%%%%%%%%%%%%%%%%%%%%%

\newpage
\appendix
\section*{Appendix}

\section{Weight Pushing for Path Sampling}
\label{sec:AppenWeiPush}
We re-visit some useful results in the literature. In this section, we consider a DAG $G$ with parameters as introduced in Section~\ref{sec:SOPPPFor}. For simplicity, we assume that each edge in $\eset$ belongs to at least one path in $\pset$. Let us respectively denote by $C(u)$ and $F(u)$ the set of the direct successors and the set of the direct predecessors of any vertex $u\in \mathcal{V}$. Moreover, let $e_{[u,v]}$ and $\pset_{u,v}$ respectively denote the edge and the set of all paths from vertex $u$ to vertex~$v$. 

Let us consider a weight $w(e)\!>\!0$ for each edge $e\in \eset$. It is needed in the \Algo/~algorithm to sample a path $\patil \in \pset$ with the probability:
\begin{equation}
x(\patil):= {\left[\prod \nolimits_{e \in \patil}{w(e)}\right]}  \Big/   \left[{\sum \nolimits_{\pa \in \pset} {\prod \nolimits_{e \in \pa}{w(e)} }}\right].\label{equAlgo2} 
\end{equation}
A direct computation and sampling from \mbox{$x(\patil), \forall \patil \in \pset$} takes $\mO(P)$ time which is very inefficient. To efficiently sample the path, we first label the vertices set by \mbox{$\mathcal{V}\!=\! \{s=u_0,u_1,\ldots, d\!=\!u_K\}$} such that if there exists an edge connecting $u_i$ to $u_j$ then $i<j$. We then define the following terms for each vertex $u \in \vset$:
\begin{equation*}
H\!( s,u )\!:=\! \sum_{\pa\! \in \pset_{s\!,u}}{\prod_{e \in \pa}{w(e)} } \textrm{ and }
H\!( u,d )\! :=\! \sum_{\pa \in \pset_{u\!,d}}{\prod_{e \in \pa}{w(e)} }.
\end{equation*}
Intuitively, $H(u,v)$ is the aggregate weight of all paths from vertex $u$ to vertex~$v$ and $H(s,d)$ is exactly the denominator in~\eqref{equAlgo2}. These terms $H(s,u)$ and $H(u,d), \forall u \in \vset$ can be recursively computed by the WP algorithm (i.e., Algorithm~\ref{dynpro}) that runs in $\mO( E )$ time, through dynamic programming. This is called \emph{weight pushing} and it is used by \cite{gyorgy2007,sakaue2018,takimoto2003}. 
%
%% %%
\begin{algorithm}[htb!]
  \caption{WP Algorithm.}
  \label{dynpro}
\begin{algorithmic}[1]
\STATE {\bfseries Input:} Graph $G$, set of weights \mbox{$\{w(e), e\in \eset\}$}.
\STATE Initialization \mbox{$H(s,u_0):=H(u_K,d):=1$}.
	\FOR{$k=1$ {\bfseries to} $K$}
    \STATE $H(u_{K-k},d) := \sum \limits_{v\in C(u_{K-k})} {w(e_{[u_{K-k},v]}) H(v,d)}$.
    \STATE $H(s,u_k) := \sum \limits_{v\in F(u_k)} {w(e_{[v,u_k]}) H(s,v)}$.
    \ENDFOR
 \STATE {\bfseries Output:} $H(s,u), H(u,d)$, $\forall u \in \vset$. 
 \end{algorithmic}
 \end{algorithm}

\begin{algorithm}[hb!]
  \caption{WPS Algorithm.}
  \label{Algo:sample}
\begin{algorithmic}[1]
\STATE {\bfseries Input:} Graph $G$, set of weights \mbox{$\{w(e), e\in \eset\}$}.
 \STATE $H(u,d), \forall u \in \mathcal{V}$ are computed by Algorithm~\ref{dynpro}.
 \STATE Initialize $\mathsf{Q} := \{ s \}$, vertex $u:= s$.
  \WHILE{$u \neq d$}
		\STATE Sample a vertex $v$ from $\mathcal{C}(u)$ with probability $w(e_{[u,v]}){H(v,d)} \big/{H(u,d)}$.
		\STATE Add $v$ to the set $\mathsf{Q}$ and update $u:= v$.
	\ENDWHILE
\STATE {\bfseries Output:} $\patil \in \pset$ going through all the vertices in~$\mathsf{Q}$
 \end{algorithmic}
\end{algorithm}

Based on the WP algorithm (i.e., Algorithm~\ref{dynpro}), we construct the WPS algorithm (i.e., Algorithm~\ref{Algo:sample}) that uses the weights $w(e),e\in \eset$ as inputs and randomly outputs a path in $\pset$. Intuitively, starting from the source vertex \mbox{$s=u_0$}, Algorithm~\ref{Algo:sample} sequentially samples vertices by vertices based on the terms $H(u,v)$ computed by Algorithm~\ref{dynpro}. It is noteworthy that Algorithm~\ref{Algo:sample} also runs in $\mO(E)$ time and it is trivial to prove that the probability that a path $\pa$ is sampled from Algorithm~\ref{Algo:sample} matches exactly $d(\pa)$.

%% %

%% %

%
% %% %
% \subsection{Lemmas on Graphs' Independence Numbers}
% In this work, we also use several lemmas in graph theory that gives the upper bounds of 

% terms of the independence number of the graphs. 
% In this section, we consider a graph $\tilde{G}$ whose vertices set and edges set are respectively denoted by $\tilde{\mathcal{V}}$ and~$\tilde{\eset}$. Let~$\tilde{\alpha}$ be its independence number\footnote{$\tilde{\alpha}$ is computed while ignoring the direction of the edges if $\tilde{G}$ is a directed graph.}.
%

%%%%%%%%%%%%%%%%%%%%%%%%%%%%%%%
%%%%%%%%%%%%%%%%%%%%%%%%%%%%%%%%
%%%%%%%%%%%%%%%%%%%%%%%%%%%
\section{Proof of Algorithm~\ref{Algo:EstLoss}'s Output}
\label{subalgo}
\begin{proof}
Fixing an edge $e\in \eset$, we prove that when Algorithm~\ref{Algo:EstLoss} takes the edges weights $\{w_t(e), e\in \eset \}$ as the input, it outputs exactly \mbox{$q_t= \sum \nolimits_{\pa \in \ob_t(e)}{x_t(\pa)} $}. We note that if \mbox{$e^{\prime} \in \mathfrak{R}_t(e):=\{e^{\prime}: e^{\prime} \rightarrow e\}$}, then \mbox{$\{\pa \in \pset: \pa \ni e^{\prime} \} \subset \ob_t(e)$}.

We denote $|\mathfrak{R}_t(e)| = \rho_e$ and label the edges in the set $\mathfrak{R}_t(e)$ by $\{e_1,e_2,\ldots, e_{\rho_e}\}$. The for-loop in lines $4$-$8$ of Algorithm~\ref{Algo:EstLoss} consecutively run with the edges in $R_t(e)$ as follows:
\begin{trivlist}
    \item[$(i)$] After the for-loop runs for $\! e_1$, we have \mbox{$K(e_1):= \sum \nolimits_{\pa \ni e_1}{\prod \nolimits_{\bar{e} \in \pa} } \bar{w}(\bar{e}) = \sum \nolimits_{\pa \ni e_1}{w_t(\pa)}$}; therefore, \mbox{$q_t(e) = \sum \nolimits_{\pa \ni e_1} x_t(\pa)$} since $H^*(s,d) = \sum \nolimits_{\pa \in \pset}{w_t(\pa)}$ computed from the original weights \mbox{$w_t(\bar{e}), \bar{e}\in \eset$}. Due to line~$8$ that sets $\bar{w}(e_1):=0$, henceforth in Algorithm~\ref{Algo:EstLoss}, the weight $\bar{w}(\pa):= \prod \nolimits_{e\in \pa} \bar{w}(e)$ of any path $\pa$ that contains $e_1$ is set to~$0$.
    \item[$(ii)$] Let the for-loop run for $e_2$, we have \mbox{$K(e_2):= \sum \nolimits_{\pa \ni e_2} {\bar{w}(\pa)} = \sum \limits_{\{\pa \ni e_2\} \backslash \{ \pa \ni e_1\}}{w_t(\pa)}$} because any path $\pa \ni e_1$ has the weight $\bar{w}(\pa)=0$. Therefore, \mbox{$q_t(e) = \sum \nolimits_{\pa \ni e_1} x_t(\pa) +  \sum \nolimits_{\{\pa \ni e_2\} \backslash \{\pa \ni e_1\}} x_t(\pa)$}.
    \item[$(iii)$] Similarly, after the for-loop runs for $e_i$ (where \mbox{$i \in \{3,\ldots, \rho_{e} \}$}), we have:
    \begin{equation*}
    q_t(e) = \sum_{k=1}^{i}{\left(  \sum \limits_{ \{\pa \ni e_k\} \backslash \bigcup \limits_{j< k}{\{\pa \ni e_j \}} } {x_t(\pa)}    \right)}.
    \end{equation*}
    \item[$(iv)$] Therefore, after the for-loop finishes running for every edge in $\mathfrak{R}_t(e)$; we have \mbox{$q_t: = \sum_{\pa \in \ob_t(e)}{x_t(\pa)}$} where each term $x_t(\pa)$ was only counted once even if $\pa$ contains more than one edge that reveals the edge~$e$.
\end{trivlist}
\end{proof}
%%%%%%%%%%%%%%%%%%%%%%%%%%
%%%%%%%%%%%%%%%%%%%%%%%%%%%
%
%
%%%%%%%%%%%%%%%%%
\section{Proof of Theorem~\ref{maintheo}}
\label{AppenTheo1}
\Theoone*
\begin{proof}
We first denote\footnote{We recall that $w_t(\pa):= \prod \nolimits_{e\in \pa} {w_t(e)}$.} \mbox{$W_t:= \sum \nolimits_{\pa \in \pset} w_t(\pa), \forall t\in[T]$}. From line $9$ of Algorithm~\ref{OEdgeAl}, we trivially~have:
\begin{align}
w_{t+1}(\pa) &= w_t(\pa) \cdot \exp(- \eta \hat{L}_t(\pa)), \forall \pa \in \pset, \forall t \in [T-1]. \label{patwei} 
\end{align}

We recall that $\hat{L_t}(\pa):= \sum \nolimits_{e \in \pa} {\hLoss_{t}(e)}$ and the notation $\mathbb{E}_t$ denoting the expectation w.r.t. to the randomness in choosing $\patil_t$ in Algorithm~\ref{OEdgeAl} (i.e., w.r.t. the information up to time $t-1$). From \eqref{optimis}, we~have:
\begin{equation}
\mathbb{E}_t\left[ \hat{L}_t(\pa) \right] \le L_t(\pa):= \sum \nolimits_{e\in \pa}{\Loss_t(e)}, \forall \pa \in \pset. \label{estiloss}
\end{equation}

Under the condition that \mbox{$0< \eta$}, we obtain:
\begin{align}
    \frac{W_{t+1}}{W_t} & = \sum \nolimits_{\pa \in \pset}{\frac{w_{t+1}(\pa)}{W_t}} \nonumber \\
    & =  \sum \nolimits_{\pa \in \pset}{\frac{w_{t}(\pa)\cdot \exp(-\eta \hat{L}_t(\pa))}{W_t}} \nonumber \\
    & = \sum \nolimits_{\pa \in \pset}{x_t(\pa) \cdot \exp(-\eta \hat{L}_t(\pa)))} \nonumber \\
    & \le \sum \limits_{\pa \in \pset} \left[x_t(\pa) \left( 1- \eta\hat{L}_t(\pa) + \frac{\eta^2}{2} (\hat{L}_t(\pa))^2\right) \right] \nonumber \\
    & = 1\! -\! \sum \limits_{\pa \in \pset} \left[x_t(\pa) \left( \eta\hat{L}_t(\pa) \! -\! \frac{\eta^2}{2} (\hat{L}_t(\pa))^2\right) \right]. \label{first}
\end{align}
Here, the second equality comes from \eqref{patwei} and the inequality comes from the fact that \mbox{$\exp(-a) \le 1 -a +a^2/2$} for \mbox{$a:=\eta\hat{L}_t(\pa) \ge 0 $}.
% \footnote{We can check that \mbox{$\eta < 2 \big/ {\frac{n}{\beta}} <2 \big/ {\sum \nolimits_{e \in p}{\frac{\Loss_t(e)}{q_t(e) + \beta}} } = 2 \big/ \hat{L}(\pa)$}, thus \mbox{$0\!\le \!\eta \hat{L}_t(\pa) \!-\!\eta^2\hat{L}_t(\pa)^2 \big/ 2$}.}
Now, we use the inequality \mbox{$\ln(1-y) \le -y$}, \mbox{$\forall y<1$} for $y:= \! \sum \nolimits_{\pa \in \pset} \left[x_t(\pa) \left( \eta\hat{L}_t(\pa) \! -\! \frac{\eta^2}{2} (\hat{L}_t(\pa))^2\right) \right]$,\footnote{We can easily check that $\eta \hat{L}_t(\pa) \!-\!\eta^2\hat{L}_t(\pa)^2/2 \! < \!1$ for any $\eta>0$ and thus, $\sum \nolimits_{\pa \in \pset} \left[x_t(\pa) \left( \eta\hat{L}_t(\pa) \! -\! \frac{\eta^2}{2} (\hat{L}_t(\pa))^2\right) \right] < 1$.} then from \eqref{first}, we obtain 
\begin{align}
     & \ln \left( \frac{W_{T+1}}{W_1}\right) \nonumber \\
     = & \sum \limits _{t=1}^T{\ln \left(\frac{W_{t+1}}{W_t} \right)} \nonumber \\
     \le & \sum \limits_{t=1}^T\!\!{\left(\!-\eta \sum \limits_{\pa \in \pset}{x_t(\pa)\hat{L}_t(\pa)}\!+\!\frac{\eta^2}{2}\sum \limits_{\pa \in \pset}{x_t(\pa) (\hat{L}_t(\pa))^2}\!\!\right)} .\label{lnWT}
\end{align}
On the other hand, let us fix a path $\pa^* \in \pset$, then
\begin{align}
    & \ln \left( \frac{W_{T+1}}{W_1}\right)  \nonumber\\
    \ge & \ln \left( \frac{w_{T+1}(\pa^*)}{W_1}\right) \nonumber \\
    = & \ln \frac{w_{T}(\pa^*) \exp(-\eta \hat{L}_T(\pa^*))}{P} \nonumber\\
    = & \ln \frac{w_{T\!-\!1}(\pa^*) \exp(-\eta \hat{L}_T(\pa^*)\!-\!\eta \hat{L}_{T-1}(\pa^*))}{P} \nonumber\\
    = & -\eta \sum \limits_{t=1}^T{\hat{L}_t(\pa^*)} - \ln(P). \label{reverse}
\end{align}
In the arguments leading to \eqref{reverse}, we again use \eqref{patwei} and the fact that \mbox{$w_1(\pa) = 1, \forall \pa\in\pset$}, including~$w_1(\pa^*)$. Therefore, combining \eqref{lnWT} and \eqref{reverse} then dividing both sides by $\eta$, we have:
\begin{align}
 & \sum \limits_{t=1}^T {\sum \limits_{\pa \in \pset} {x_t(\pa) \hat{L}_t(\pa)}}   \nonumber \\
 \le &  \frac{\ln(P)}{\eta} + \sum \limits_{t=1}^T{\hat{L}_t(\pa^*)} +  \frac{\eta}{2}  \sum \limits_{t=1}^T {\!\sum \limits_{\pa \in \pset}{x_t(\pa) (\hat{L}_t(\pa))^2}}. \label{be4exp}
\end{align}

Now, we take $\mathbb{E}_t$ on both sides of \eqref{be4exp}, then we apply \eqref{estiloss} to obtain:
\begin{align}
 & \sum \limits_{t=1}^T{\sum \limits_{\pa \in \pset} {x_t(\pa) \mathbb{E}_t[\hat{L}_t(\pa)]} }  \nonumber\\
    \le &  \frac{\ln(P)}{\eta} \! + \!\sum \limits_{t=1}^T{{L}_t(\pa^*)} \! + \! \frac{\eta}{2}\sum \limits_{t=1}^T{\sum \limits_{\pa \in \pset}{x_t(\pa) \mathbb{E}_t[\hat{L}_t(\pa)^2]}}. \label{afterexp}
\end{align}

Now, we look for a lower bound of \mbox{$\sum \nolimits_{\pa \in \pset} {x_t(\pa) \mathbb{E}_t\left[\hat{L}_t(\pa)\right]}$}.
% To ease the notation, for each edge $e \in \eset$, let us denote
% \begin{equation}
% r_t(e)\!:=\!\sum \limits_{\pa^{\prime} \ni e} {x_t(\pa^{\prime})} \textrm{ and } q_t(e)\!:=\!\sum \limits_{\pa^{\prime} \in \ob_t(e)} {x_t(\pa^{\prime})}. \label{defnotation}
% \end{equation}
For any fixed \mbox{$\pa \in \pset$}, we consider:
\begin{align}
\mathbb{E}_t\left[\sum \limits_{e \in \pa} \hLoss_t(e) \right] = & \sum \limits_{\patil \in \pset} {\left[ x_t(\patil) \sum \limits_{e\in \pa} \left( \frac{\Loss_t(e)}{q_t(e)\!+\!\beta} \mathbb{I}_{\left\{e \in \ob_t(\patil) \right\}}\right) \right]} \nonumber \\
= & \sum \limits_{e\in \pa} {\sum \limits_{\patil \in \mathbb{O}(e)}{x_t(\patil) \frac{\Loss_t(e)}{q_t(e) + \beta}}} \nonumber \\
= &  \sum \limits_{e \in \pa} {\frac{q_t(e) \Loss_t(e)}{q_t(e) + \beta}}. \label{lowbound1}
\end{align}

Using \eqref{lowbound1} and recalling that \mbox{$\Loss_t(e) \le 1, \forall e\in\eset$}, we~have:
\begin{align}
    & \sum \limits_{\pa \in \pset} {x_t(\pa) \mathbb{E}_t\left[\hat{L}_t(\pa)\right]} - \sum \limits_{\pa \in \pset} {x_t(\pa) L_t(\pa)}  \nonumber\\
    = & \sum \limits_{\pa \in \pset} {x_t(\pa) \sum \limits_{e \in \pa} {\frac{q_t(e) \Loss_t(e)}{q_t(e) + \beta}}}  - \sum \limits_{\pa \in \pset} {x_t(\pa) \sum \limits_{e \in \pa}{ \Loss_t(e)}} \nonumber \\
    = & \sum \limits_{\pa \in \pset} {x_t(\pa) \sum \limits_{e \in \pa} {\Loss_t(e)\left( \frac{q_t(e)}{q_t(e) + \beta} -1\right)}} \nonumber \\
    \ge & - \sum \limits_{\pa \in \pset} {x_t(\pa) \sum \limits_{e \in \pa} { \frac{\beta}{q_t(e) + \beta}}} \nonumber \\
    = & - \beta \sum \limits_{e \in \eset} { \frac{\sum \limits_{\pa \ni e} {x_t(\pa)}}{q_t(e) + \beta}} \nonumber \\
    = & - \beta Q_t . \label{lowerbound}
\end{align}
Therefore, a lower bound of \mbox{$\sum \nolimits_{\pa \in \pset} {x_t(\pa) \mathbb{E}_t\left[\hat{L}_t(\pa)\right]}$} is \mbox{$\sum \nolimits_{\pa \in \pset} {x_t(\pa) L_t(\pa)} - \beta Q_t  $}.

Now, we look for an upper bound of \mbox{$\sum \nolimits_{\pa \in \pset}{x_t(\pa) \mathbb{E}_t\left[\hat{L}_t(\pa)^2\right]}$}. To do this, fix $\pa \in \pset$, we~consider 
\begin{align}
    & \mathbb{E}_t \left[\hat{L}_t(\pa ) ^2 \right]  \nonumber \\
    = & \mathbb{E}_t \left[\left(\sum \nolimits_{e\in \pa}{\hLoss_t(e)} \right)^2  \right]  \nonumber \\
    \le & n \cdot  \mathbb{E}_t \left[\sum \nolimits_{e\in \pa}{\hLoss_t(e)^2}   \right] \nonumber \\
    = & n \cdot  \sum \limits_{\patil \in \pset} {\left[x_t(\patil) \sum \limits_{e\in \pa}{ \left(\frac{\Loss_t(e)}{q_t(e)+ \beta}   \mathbb{I}_{\{e \in \ob_t(\patil)\}} \right)^2 } \right] }  \nonumber \\
    \le & n \cdot  \sum \limits_{e \in \pa} {\sum \limits_{\patil \in \ob_t(e)} {x_t(\patil) \frac{1}{(q_t(e)+ \beta)^2}  } } \nonumber \\
    = & n \cdot  \sum \limits_{e \in \pa} { q_t(e) \frac{1}{(q_t(e) + \beta)^2}  } \nonumber \\
    \le & n \cdot  \sum \limits_{e \in \pa} {  \frac{1}{q_t(e) + \beta}  } .
    \label{upperbound1}
\end{align}
The first inequality comes from applying Cauchy–Schwarz inequality. The second inequality comes from the fact that \mbox{$\Loss_t(e)\le 1$} and the last inequality comes from \mbox{$q_t(e) \le q_t(e)+ \beta$ since $\beta>0$}. 

Now, applying \eqref{upperbound1}, we can bound
\begin{align}
\sum \limits_{\pa \in \pset}{x_t(\pa) \mathbb{E}_t\left[\hat{L}_t(\pa)^2\right]} \le & n \cdot \sum \limits_{\pa \in \pset}{x_t(\pa) \sum \limits_{e \in \pa} {  \frac{1}{q_t(e) + \beta}  } } \nonumber \\
= & n \cdot \sum \limits_{e\in \eset}{ \sum \limits_{\pa \ni e} x_t(\pa) {\frac{1}{q_t(e)+\beta} }} \nonumber \\
= & n \cdot \sum \limits_{e\in \eset} {\frac{r_t(e)}{q_t(e)+\beta}}
= n \cdot Q_t.\label{upperbound}
\end{align}
Here, we recall the notation $r_t(e)$ and $Q_t$ defined in Section~\ref{sec:OEPerform}. Replacing \eqref{lowerbound} and \eqref{upperbound} into \eqref{afterexp}, we have that the following inequality holds for any \mbox{$\pa^* \in \pset$}.
\begin{align}
    & \sum \limits_{t=1}^T{\sum \limits_{\pa\in \pset}{x_t(\pa) L_t(\pa)}}  - \sum\limits_{t=1}^T{\beta Q_t} - \sum \limits_{t=1}^T{L_t(\pa^*)}  \nonumber\\
    \le  & \frac{\ln(P)}{\eta}\!+\! \frac{\eta}{2} \sum \limits_{t=1}^T{n Q_t}\nonumber.
\end{align}
Therefore, we conclude that 
\begin{align*}
    R_T & =  \sum \limits_{t=1}^T{\sum \limits_{\pa\in \pset}{x_t(\pa) L_t(\pa)}}   - \sum \limits_{t=1}^T{L_t(\pa^*)} \\
    & \le  \frac{\ln(P)}{\eta} + \sum \limits_{t=1}^T{Q_t\left(n\frac{\eta}{2} + \beta \right) }.
\end{align*}

\end{proof}

%%%%%%%%%%%%%%%%%%%%
%%%%%%%%%%%%%%%%%%%%%
%%%%%%%%%%%%%%%%%%%%%%%%
\section{Lemmas on Graphs' Independence Numbers}
\label{graphlemma}
In this section, we present some lemmas in graph theory that will be used in the next section to prove Theorem~\ref{theoremQ}. Consider a graph $\tilde{G}$ whose vertices set and edges set are respectively denoted by $\tilde{\mathcal{V}}$ and~$\tilde{\eset}$. Let~$\tilde{\alpha}$ be its independence~number.
\begin{lemma}
\label{lemG1}
Let $\tilde{G}$ be an directed graph and $I_{v}$ be the in-degree of the vertex $v \in \tilde{\mathcal{V}}$, then 
\begin{equation}
\sum \nolimits_{v \in \tilde{\mathcal{V}}}\left[{1} / {(1+I_{v})}\right] \le 2 {\tilde{\alpha}}\ln\left(1+ {|\tilde{\mathcal{V}}|}/{{\tilde{\alpha}}} \right).\nonumber
\end{equation}
\end{lemma}
%
%% %%
A proof of this lemma can be found in Lemma 10 of \cite{alon13}.
\begin{lemma}
\label{lemG2}
Let $\tilde{G}$ be a directed graph with self-loops and consider the numbers \mbox{$k(v)\in [0,1],\forall v\in \tilde{\mathcal{V}}$} such that there exists $\gamma>0$ and $\sum \nolimits_{v\in \tilde{\mathcal{V}}} {k(v)} \le \gamma$. For any~$c>0$,~we have
\begin{equation*}
\sum \limits_{v\in \tilde{\mathcal{V}}}{ \frac{k(v)}{\frac{1}{\gamma}\!\sum \limits_{v^{\prime} \rightarrow v}\!{k(v^{\prime})\!+\! c} }} \le 2\gamma\tilde{\alpha}\ln\!\left(\!1\!+\!\frac{\gamma\lceil|\tilde{\mathcal{V}}|^2/c \rceil +  |\tilde{\vset}| }{\tilde{\alpha}}\! \right)\! +\! 2\gamma.
\end{equation*}
\end{lemma}
%
%% %%
A proof of this lemma can be found in Lemma 1 of \cite{kocak14}.
\begin{lemma}
\label{lemG3}
Let $\tilde{G}$ be an undirected graph with self-loops and consider the numbers \mbox{$k(v) \ge  0$, $v \in  \tilde{\mathcal{V}}$}. We have
\begin{equation*}
\sum \nolimits_{v\in \tilde{\mathcal{V}}} \left[{k(v)}\big/{\sum \nolimits_{v^{\prime} \rightarrow v}\!{k(v^{\prime})} }\right] \le \tilde{\alpha}. 
\end{equation*}
\end{lemma}
%
%% %%
This lemma is extracted from Lemma $3$ of \cite{mannor11}.

%%%%%%%%%%%%%%%%%%%%%%%%%%%%%%%%%%%%%%%
%%%%%%%%%%%%%%%%%%%%%%%%%%%%%%%%%%%%%%
%%%%%%%%%%%%%%%%%%%%%%%%%%%%%%%%%%%%%
\section{Proof of Theorem~\ref{theoremQ}}
\theoremQ*
\label{sec:AppenTheo2}
\textbf{{Case 1:}} \emph{$G^O_t$ does not satisfy assumption~$(A0)$}. Fixing an edge $e$, due to the fact that $n$ is the length of the longest paths in $\pset$, we have
\begin{align}
& n q_t(e) \! =\! n\!\!\sum \limits_{ \pa \in \ob_t(e)}{x_t(\pa)} \ge \sum \limits_{e^{\prime} \rightarrow e}{ \sum \limits_{\pa \ni e^{\prime}}{x_t(\pa)}\!}\!=\! \sum \limits_{e^{\prime} \rightarrow e}{r_t(e^{\prime})} \nonumber \\
\Rightarrow & Q_t \!=\!\sum \limits_{e\in \eset} {\frac{r_t(e)}{q_t(e)\!+\!\beta} \!} \le  \sum \limits_{e\in \eset} {\frac{r_t(e)}{\frac{1}{n}\! \sum \limits_{e^{\prime} \rightarrow e}{r_t(e^{\prime})} \! + \! \beta} }. \label{Qtbound}
\end{align}
% This inequality implies that
% \begin{equation}
% Q_t \!=\!\sum \limits_{e\in \eset} {\!\!\frac{r_t(e)}{q_t(e)\!+\!\beta} \!} \le  \sum \limits_{e\in \eset} {\frac{r_t(e)}{\frac{1}{n}\! \sum \limits_{e^{\prime} \rightarrow e}{r_t(e^{\prime})} \! + \! \beta} }. \label{Qtbound}
% \end{equation}
%
%
\emph{Case 1.1:} If $G^O_t$ is a non-symmetric (i.e., directed) graph, we apply Lemma \ref{lemG2} with $\gamma =n, c= \beta$ on the graph $\tilde{G} = G^O_t$ (whose vertices set $\tilde{\mathcal{V}}$ corresponds to the edges set $\eset$ of $G$) and the numbers\footnote{We verify that these numbers satisfy
\begin{equation*}
\sum_{e\in \eset}{r_t(e)}\! =\! \sum_{e\in \eset}{\!\sum \limits_{\pa \ni e}x_t(\pa)} \! = \! \sum_{\pa \in \pset} {\!\sum_{e\in \pa}{x_t(\pa)}}\!\le\! \! \sum_{\pa \in \pset} {n{x_t(\pa)}} = n.
\end{equation*}
}
\mbox{$ k(v_e) = {r}_t(e), \forall v_e \in \tilde{\vset}$} (i.e., $\forall e \in \eset$). 
We obtain the following inequality:
\begin{equation*}
 \sum \limits_{e\in \eset} { \frac{r_t(e)}{\frac{1}{n}\!\sum \limits_{e^{\prime} \rightarrow e}{r_t(e^{\prime})} \! + \! \beta}  }  \le 2n\alpha_t\ln\left(1\!+\!\frac{n\lceil E ^2/\beta \rceil\! +\! E}{\alpha_t} \right) + 2n.
\end{equation*}
\emph{Case 1.2:} If $G^O_t$ is a symmetric (i.e. undirected) graph, we apply Lemma~\ref{lemG3} with the graph \mbox{$\tilde{G} = G^O_t$} (whose vertices set $\tilde{\mathcal{V}}$ corresponds to the edges set $\eset$ of the graph $G$) and the numbers \mbox{$k(v_e) = r_t(e), \forall v_e \in \tilde{V}$} (i.e., $\forall e \in \eset$) to obtain:
\begin{equation*}
\sum \limits_{e\in \eset} {\frac{r_t(e)}{\frac{1}{n}\sum \limits_{e^{\prime} \rightarrow e}{r_t(e^{\prime})\!+\!\beta}} }  \le n \sum \limits_{e\in \eset} {\frac{r_t(e)}{\sum \limits_{e^{\prime} \rightarrow e}{r_t(e^{\prime})}} }  \le n \alpha_t. \label{on}
\end{equation*}
\textbf{{Case 2:}} \emph{$G^O_t$ satisfies assumption $(A0)$}. 
% \begin{equation}
% \sum \nolimits_{ \pa \in \ob_t(e)}{x_t(\pa)} \le \sum \nolimits_{e^{\prime}\!\rightarrow e}{\left(\sum \nolimits_{\pa \ni e^{\prime}}{x_t(\pa)}\right)}. \label{inequpperqt}
% \end{equation}
%
% In other words, \mbox{$q_t(e ) \le \sum \nolimits_{e^{\prime}\! \rightarrow e}{r_t(e^{\prime})}$}. This is due to the fact that each term $x_t(\pa),\pa \in \ob_t(e)$ appears only one time in the LHS but may appear more than once in the RHS of~\eqref{inequpperqt}. However, the summations in the RHS of \eqref{inequpperqt} contain no more than $n$ times the term \mbox{$x_t(\pa)$} for each $\pa \in \ob_t(e)$ since the length of the longest path is $n$; 
Under this assumption, $q_t(e) = \sum \nolimits_{e^{\prime}\rightarrow e}{r_t(e^{\prime})}$ due to the definition of~\mbox{$\ob_t(e)$}. Therefore,
\mbox{$Q_t =  \sum \nolimits_{e\in \eset} {\left[{r_t(e)}\big/\left({\sum \nolimits_{e^{\prime} \rightarrow e}{r_t(e^{\prime})}  +\beta} \right) \right]}$}.

\emph{Case 2.1:} If $G^O_t$ is a non-symmetric (i.e., directed) graph. We consider a discretized version of $x_t(\pa)$ for any path $\pa \in \pset$ that is $\tilde{x}_t(\pa) := k/M$ where $k$ is the unique integer such that \mbox{$(k-1)/M \le x_t(\pa) \le k/M $}; thus, \mbox{$\tilde{x}_t(\pa) - 1/M \le x_t(\pa) \le \tilde{x}_t(\pa)$}.

Let us denote the discretized version of $r_t(e)$ by \mbox{$\tilde{r}_t(e) := \sum \nolimits_{\pa \ni e}{\tilde{x}_t(\pa)}$}. We deduce that \mbox{$r_t(e) \le \tilde{r}_t(e)$}~and 
\begin{equation*}
\sum \limits_{e^{\prime} \rightarrow e}{{r}_t(e)} \ge \sum \limits_{e^{\prime} \rightarrow e}{\left( \tilde{r}_t(e^{\prime}) - \frac{1}{M} \right)}
\ge \sum \limits_{e^{\prime} \rightarrow e}{\tilde{r}_t(e^{\prime})} - \frac{ E }{M}.
\end{equation*}
We obtain the bound:
\begin{align}
Q_t =  \sum \limits_{e\in \eset} \frac{r_t(e)}{ \left({\sum \limits_{e^{\prime} \rightarrow e}{r_t(e^{\prime})}  +\beta} \right)} \le  \sum \limits_{e\in \eset} {\frac{\tilde{r}_t(e)}{ \sum \limits_{e^{\prime} \rightarrow e} {\tilde{r}_t(e^{\prime})}\!- \! E/M\!+\! \beta} } .
\label{be4inequ}
\end{align}
We now consider the following inequality: If $a,b \ge 0$ and \mbox{$a+b \ge B > A >0$}, then 
\begin{equation}
\frac{a}{a+b-A} \le \frac{a}{a+b} + \frac{A}{B-A}. \label{inequ}
\end{equation}
A proof of this inequality can be found in Lemma~12 of \cite{alon13}. Applying~\eqref{inequ}\footnote{Trivially, we can verify that $a+b \ge B$ and $B > A$ comes from the fact that 
$ \beta \ge \beta \frac{ 1 }{E} > \frac{ E }{ \lceil 2 E^2 / \beta \rceil}$.
}
 with $a = \tilde{r}_t(e)$, \mbox{$b=\sum \limits_{e^{\prime} \rightarrow e, e^{\prime} \neq e}{\tilde{r}_t(e^{\prime})} \! + \! \beta$}, \mbox{$A=  \frac{E}{M}$}, and \mbox{$B=\beta$} to \eqref{be4inequ},
\begin{align}
    Q_t \le  & \sum \limits_{e\in \eset} {\left(\frac{\tilde{r}_t(e)}{\sum \limits_{e^{\prime} \rightarrow e} {\tilde{r}_t(e^{\prime})} + \beta} + \frac{ E /M}{\beta -  E /M}  \right) } \nonumber\\
    \le & \sum \limits_{e\in \eset} {\frac{\tilde{r}_t(e)} {\sum \limits_{e^{\prime} \rightarrow e}\tilde{r}_t(e^{\prime})}} + 1. \label{afterinequ}
\end{align}
The last inequality comes from the fact that \mbox{$\frac{ E }{M\beta -  E } \le \frac{ E }{2 E ^2-  E } \le \frac{1}{2 E -1} \le \frac{1}{E}, \forall E \ge 1$}. 

Finally, we create an auxiliary graph ${G}^*_t$ such that:
\begin{trivlist}
    \item[$(i)$] Corresponding to each edge $e$ in $G$ (i.e., each vertex $v_e$ in $G^O_t$), there is a clique, called $\mathbb{C}(e)$, in the auxiliary graph ${G}^*_t$ with $M\tilde{r}_t(e)\in \mathbb{N}$ vertices.
    \item[$(ii)$] In each clique $\mathbb{C}(e)$ of ${G}^*_t$, all vertices are pairwise connected with length-two cycles. That is, for any \mbox{$k, k^{\prime} \in \mathbb{C}(e)$}, there is an edge from $k$ to $k^{\prime}$ and there is an edge from $k^{\prime}$ to $k$ in ${G}^*_t$.
    \item[$(iii)$] If \mbox{$e \rightarrow e^{\prime}$}, i.e., there is an edge in $G^O_t$ connecting $v_e$ and $v_{e^{\prime}}$; then in ${G}^*_t$, all vertices in the clique $\mathbb{C}(e)$ are connected to all vertices in $\mathbb{C}(e^{\prime})$.
\end{trivlist}

We observe that the independence number $\alpha_t$ of $G^O_t$ is equal to the independence number of ${G}^*_t$. Moreover, the in-degree of each vertex $k \in \mathbb(e)$ in the graph ${G}^*_t$ is:
\begin{equation}
{I}^*_k = M\tilde{r}_t(e)\! - \!1 \!+\! \sum \limits_{e^{\prime} \rightarrow e, e^{\prime} \neq e}{M \tilde{r}_t(e^{\prime})}\! = \!\sum \limits_{e^{\prime} \rightarrow e} \!{M\tilde{r}_t(e^{\prime}) }\! -\!1. \label{blabla}
\end{equation}
Let us denote ${V}^*_t$ the set of all vertices in ${G}^*_t$, then we have:
\begin{align}
& \sum \limits_{e \in \eset}\frac{\tilde{r}_t(e)} {\sum \limits_{e^{\prime} \rightarrow e} {\tilde{r}_t(e^{\prime}) }}\! =\! \sum \limits_{e \in \eset}\frac{M\tilde{r}_t(e)} {\sum \limits_{e^{\prime} \rightarrow e} {M\tilde{r}_t(e^{\prime}) }} \! =\! \sum \limits_{e \in \eset} {\sum \limits_{k\in \mathbb{C}(e)} {\frac{1}{{I}^*_k \!+ \! 1}}}\nonumber \\
	= &  \sum \limits_{k \in {V}^*_t} \frac{1}{\tilde{I}_k + 1} \le 2\alpha_t \ln\left( 1+ \frac{M+ E }{\alpha_t}\right).\label{finalstep}
\end{align}
Here, the second equality comes from the fact that $|\mathbb{C}(e)| = M \tilde{r}_t(e)$ and \eqref{blabla}. The inequality is obtained by applying Lemma~\ref{lemG1} to the graph ${G}^*_t$ and the fact that \mbox{$|{V}^*_t| = \sum \nolimits_{e\in \eset}{M\tilde{r}_t(e)} \le  M \sum \nolimits_{e\in \eset}{({r}_t(e)\! +\! 1/M)} \! \le \!  E\! +\! M $}.

In conclusion, combining \eqref{afterinequ} and \eqref{finalstep}, we obtain the regret-upper bound as given in Theorem~\ref{theoremQ} for this case of the observation graph.

%%%%%%%%%%%%%%%%%%%%%%%%%%%%%%%%%
\emph{Case 2.2:} Finally, if $G^O_t$ is a symmetric (i.e., undirected) graph, we again apply Lemma~\ref{lemG3} to the graph \mbox{$\tilde{G} = G^O_t$} and the numbers \mbox{$k(v_e) = r_t(e)$} to obtain that
\mbox{$Q_t \le  \sum \nolimits_{e\in \eset} {\left[{r_t(e)}\big/{\sum \nolimits_{e^{\prime} \rightarrow e}{r_t(e^{\prime})} } \right]} \le \alpha_t$}.\qed
\section{Parameters Tuning for \Algo/: Proof of Corollary~\ref{corrol1}}
\label{AppenTheo3}
In this section, we suggest a choice of $\beta$ and $\eta$ that guarantees the expected regret given in Corollary~\ref{corrol1}.
\coroll*
\textbf{Case 1: Non-symmetric (i.e. directed) observation graphs that do not satisfy assumption~$(A0)$}. We find the parameters $\beta$ and $\eta$ such that $R_t \le \tilde{\mO}\left(n\sqrt{T \alpha} \right)$. We note that \mbox{$\alpha_t \ge 1$}, $\forall t \in [T]$; therefore, recalling that $\alpha$ is an upper bound of $\alpha_t$, from Theorem~\ref{maintheo} and \ref{theoremQ}, we have:
\begin{align}
    R_T &  \le  \frac{\ln(P)}{\eta} \! + \! \sum_{t=1}^T{ \left(\! n\frac{\eta}{2}\!+\!\beta\! \right) 2n\left[1\!\!+\!{\alpha_{t}} \ln \! \left(1\! +\! \frac{nM\!+ \!E}{\alpha_t} \! \right) \!  \right]} \nonumber \\ 
    & \le \frac{\ln(P)}{\eta} \! + \! T{\left(n\frac{\eta}{2}\!+\!\beta \right) 2n \left[1+ \alpha\ln\left(\alpha \! +\!{nM}\! + \! E  \right) \right]} \nonumber \\
    & = \frac{\ln(P)}{\eta} + \eta T{n^2}\left[1+ \alpha\ln\left(\alpha +{nM}+  E   \right) \right] \nonumber\\
        & \qquad  + 2\beta Tn \left[1+ \alpha\ln\left(\alpha+{nM} +  E  \right) \right] \label{eq:tuning1}.
\end{align}
Recalling that $M:= \lceil{ 2 E^2/\beta }\rceil$, by choosing any 
\begin{align}
    \beta & \le {1}/{\sqrt{Tn[1+ \alpha \ln(\alpha+n \lceil { E ^2}/{\beta}\rceil+  E  )]}}, \label{eq:tune_beta1}\\
 \textrm{and }   \eta & = {\sqrt{\ln(P)}}/{\sqrt{n^2T\left[ 1+ \alpha  \ln\left(\alpha + n \lceil{ E ^2}/{\beta}  \rceil  +   E  \right) \right] }} \nonumber,
\end{align}
we obtain the bound:
\begin{align}
R_T  \le & 2 n\sqrt{T \ln(P) \cdot [1+ \alpha \ln(\alpha+nM+  E  )] }  \nonumber\\
    & \qquad + 2 \sqrt{ Tn[\alpha+ \alpha \ln(\alpha+nM+  E  )] } \label{eq:RTBOU} \\
\le & \tilde{\mO}\left(n \sqrt{T \alpha\ln(P)}  \right). \nonumber
\end{align}

In practice, as long as it satisfies ~\eqref{eq:tune_beta1}, the larger $\beta$ is, the better upper-bounds that \Algo/ gives. As an example that \eqref{eq:tune_beta1} always has at least one solution, we now prove that it holds with
\begin{equation}
\beta^* = \frac{-T n^2  E ^2\!+\! \sqrt{(Tn^2E^2)^2\! +\! 4Tn(1 \!+\! \alpha\ln{\alpha} 
\!+\!  E\!  +\!  n)}}{2Tn(1\!+\! \alpha\ln{\alpha} \! +\!  E\!  +\! n) }. \label{eq:beta^*}    
\end{equation} 
Indeed, $\beta^*>0$ and it satisfies: 
\begin{align*}
& {\beta^*}^{2}\cdot Tn (1+ \alpha \ln {\alpha} +  E  + n) + \beta^* Tn^2 E ^2 = 1. \\
\Rightarrow &  {\beta^*}^{2}\cdot Tn (1+ \alpha \ln {\alpha}+  E   ) + {\beta^*}^2 Tn^2\left(  \frac{ E ^2}{\beta^*} + 1 \right)= 1\\
\Rightarrow & {\beta^*}^{2}\cdot Tn (1+ \alpha \ln {\alpha} +  E  )  + {\beta^*}^2 Tn^2 \ceil[\Big]{\frac{ E ^2}{\beta^*}} \le 1 \\
\Rightarrow & \beta^* \le \frac{1}{\sqrt{Tn\left( 1 + \alpha \ln{\alpha} +  E  + nM \right)}}.
\end{align*}

On the other hand, applying the inequality $\ln(1+x) \le x$, $\forall x \ge 0$, we have:
\begin{align*}
& \frac{nM+  E  }{\alpha}  \ge \ln\left(1+ \frac{nM+  E  }{\alpha} \right) \\
 \Rightarrow & \frac{nM+  E  }{\alpha} + \ln\alpha \ge \ln(\alpha + nM+  E  ) \\
\Rightarrow &nM +  E  + \alpha \ln\alpha + 1 \ge \alpha \ln(\alpha+nM+  E  ) + 1 \\
\Rightarrow &  \frac{1}{\sqrt{Tn\left(\! 1\! +\! \alpha \! \ln{\alpha} \!+\! nM\! +\!\! E\!  \right)}} \! \le \!  \frac{1}{\sqrt{Tn\left(\! \alpha\! \ln{(\!\alpha\! +\! nM\!+\!  E\!  )\! +\!\! 1}  \!\right)}}.
\end{align*}

Therefore, $\beta^*$ satisfies \eqref{eq:tune_beta1}. Finally, note that with the choice of $\beta = \beta^* = \Omega \left(nE^2 / [1 \!+\! \alpha\ln \alpha \!+\! E\! +\! n] \right)$ as in~\eqref{eq:beta^*}, we have
\begin{equation*}
    M = \lceil 2 E^2 / \beta \rceil \le \mO([1 \!+\! \alpha\ln \alpha \!+\! E\! +\! n] / n).
\end{equation*}
Combining this with~\eqref{eq:RTBOU}, we obtain the regret bound indicated in Section~\ref{sec:OEPerform}.

%
%

 %%%%%%%%%%%%%%%%%%%%%%%%%
 %%%%%%%%%%%%%%%%%%%%%%%%%
\textbf{Case 2: symmetric observation graphs that do not satisfy $(A0)$}. Trivially, we have that if \mbox{$\beta:= 1/\sqrt{n \alpha T}$} and \mbox{ $\eta = 2 \sqrt{\ln(P)}/\sqrt{n^2 \alpha T}$}, then
\begin{align}
R_T &\le \frac{\ln(P)}{\eta} + \left( n \frac{\eta}{2} + \beta \right) n \alpha T \nonumber \\
& = \frac{1}{2} n \sqrt{\alpha T \ln(P)}  + n \sqrt{\alpha T \ln(P)} + \sqrt{n \alpha T} \label{RTHS1}\\
& \le \tilde{\mO} \left(n \sqrt{\alpha T \ln(P)} \right). \nonumber
\end{align}

%%%%%%%%%%%%%%%%%%%%%%%%%%%%%%%%%
%%%%%%%%%%%%%%%%%%%%%%%%%%%%%%%%%

\textbf{Case 3: non-symmetric observation graphs $G^O_t$ satisfying assumption $(A0)$, $\forall t$}. We will prove that \mbox{$R_T \le 
 \tilde{\mO}\left(\sqrt{nT\alpha \ln(P)} \right)$} for any
\begin{align}
     \beta & \le {1} / {\sqrt{T\alpha[1+ 2\ln \left( 1 + \lceil  E ^2/\beta \rceil +  E  \right)]}}, \label{{eq:beta2}} \\ 
     \eta & = {2\sqrt{\ln(P)}}/{\sqrt{Tn\alpha\left[1+ 2\ln\left({\alpha} + M +  E  \right)\right] }}. \label{{eq:eta2}}
\end{align}
 
Indeed, from Theorem \ref{maintheo} and~\ref{theoremQ}, we have:
\begin{align}
    R_T & \le \frac{\ln(P)}{\eta} + \sum \limits_{t=1}^T{\!\left(n\frac{\eta}{2}\!+\!\beta\! \right)\!\left[\!1\!+\! 2\alpha_t\ln\!\left(1\!+\!\frac{M \!+\! E }{\alpha_t}\! \right)\! \right]} \nonumber \\
    & \le \frac{\ln(P)}{\eta} + \sum \limits_{t=1}^T{\left(n\frac{\eta}{2}\!+\!\beta \right)\left[\alpha + 2\alpha\ln\left(1+ M+ E  \right) \right]} \nonumber \\
    & = \frac{\ln(P)}{\eta} + \eta T \alpha \frac{n}{2} \left[1+ 2\ln\left(1+M+ E  \right) \right]\nonumber \\
        & \qquad + \beta T \alpha \left[1+ 2\ln\left(1+M +  E  \right) \right]. \label{{eq:tuning2}} 
\end{align}
We replace \eqref{{eq:beta2}} and \eqref{{eq:eta2}} into \eqref{{eq:tuning2}} and obtain:
\begin{align}
R_T & \le \frac{3}{2} \sqrt{Tn\alpha\left[1+ 2\ln\left(1 + M +  E  \right)\right]\cdot \ln(P)} \nonumber\\
    & \qquad + \sqrt{T\alpha\left[1+ 2\ln\left(1 + M +  E  \right)\right]} \label{RTCB}. \\
& \le \tilde{\mO}\left( \sqrt{n \alpha T \ln(P)}\right). \nonumber
\end{align}
A choice for $\beta$ that satisfies~\eqref{{eq:beta2}} is
\begin{equation}
\beta^*: =\frac{  -T \alpha E^2\! +\! \sqrt{(T \alpha E^2)^2 \!+\! T \alpha (3+2E)} }{T \alpha (3 +2E)}. \label{eq:beta^*2}
\end{equation}
Moreover, with this choice of $\beta^* = \Omega(E^2/(3+2E))$, we can deduce that \mbox{$M:= \lceil 2E^2/ \beta^* \rceil \le \mO(3 + 2E)$}. Combining this with~\eqref{RTCB}, we obtain the regret bound indicated in Section~\ref{sec:OEPerform}.

%%%%%%%%%%%%%%%%%
%%%%%%%%%%%%%%
\textbf{Case 4: all observation graphs are symmetric and satisfy  $(A0)$}. From Theorem~\ref{maintheo} and \ref{theoremQ}, we trivially have that if $\beta:= 1/\sqrt{\alpha T}$ and \mbox{ $\eta = 2 \sqrt{\ln(P)}/\sqrt{n \alpha T}$}, then \mbox{$R_T \le 2\sqrt{n \alpha T \ln(P)} + \sqrt{\alpha T} \le \tilde{\mO}\left( \sqrt{n \alpha T \ln(P)} \right)$}.

%%%%%%%%%%%%%%%%%%%
%%%%%%%%%%%%%%%%%%%%%%
%%%%%%%%%%%%%%%%%%%%%%

%%%%%%%%%%%%%%%%%%%%%%%%%%%%%%%
%%%%%%%%%%%%%%%%%%%%%%%%%%%%%%%%
%%%%%%%%%%%%%%%%%%%%%%%%%%%
\section{Graphical Representation of the Games' Actions Sets}
\label{graph}
\subsection{The Actions Set of the Colonel Blotto Games}
We give a description of the graph corresponding to the actions set of the learner in the CB game who distributes $k$ troops to $n$ battlefields. 
\begin{definition}[CB Graph]
    The graph $G_{k,n}$ is a DAG that~contains:
    \begin{trivlist}
    \item[$(i)$] \mbox{$N:=2+(k+1)(n-1)$} vertices arranged into $n+1$ layers.  Layer $0$ and Layer $n$, each contains only one vertex, respectively labeled $s:=(0,0)$--the source vertex and $d:=(n,k)$--the destination vertex. Each Layer \mbox{$i \in [n-1]$} contains \mbox{$k+1$} vertices whose labels are ordered from left to right by \mbox{$(i,0),(i,1),\ldots,(i,k)$}.
    \item[$(ii)$] There are directed edges from vertex $(0,0)$ to every vertex in Layer $1$ and edges from every vertex in Layer $n-1$ to vertex $(n,k)$. For $i \in \{1,2,\ldots, n-2 \}$, there exists an edge connecting vertex $(i,j_1)$ (of Layer $i$) to vertex $(i+1,j_2)$ (of Layer $(i+1)$) if $k \ge j_2 \ge j_1 \ge 0$.
    \end{trivlist}
\end{definition}
Particularly, $G_{k,n}$ has {$E=\! (k\!+\!1)\left[4\! +\! (n\!-\!2)(\!k\!+\!2) \right]\!/2 = \mO(nk^2)$} edges and \mbox{$P =  \binom{n+k-1}{n-1} = \mO(2^{\min\{n-1,k\}})$} paths going from vertex $s:=(0,0)$ to vertex $d:=(k,n)$. The edge connecting vertex $(i,j_1)$ to vertex $(i+1,j_2)$ for any $i \in \{0,1,\ldots,n-1 \}$ represents allocating $(j_2 - j_1)$ troops to battlefield $i+1$. Moreover, each path from $s$ to $d$ represents a strategy in~$S_{k,n}$. This is formally stated in Proposition~\ref{Propolayer}.
\begin{proposition}
    \label{Propolayer}
    Given $k$ and $n$, there is a one-to-one mapping between the action set $S_{k,n}$ of the learner in the CB game (with $k$ troops and $n$ battlefields) and the set of all paths from vertex $s$ to vertex $d$ of the graph $G_{k,n}$.
\end{proposition}
The proof of this proposition is trivial and can be intuitively seen in Figure \ref{fig1}-(a). We note that a similar graph is studied by \cite{Behnezhad17a}; however, it is used for a completely different purpose and it also contains more edges and paths than $G_{k,n}$ (that are not useful in this work).

%%%%%%%%%%%%%%%%%%%%%%%%%%%%%%%%%%%%%%%%
\subsection{The Actions Set of the Hide-and-Seek game}
We give a description of the graph corresponding to the actions set of the learner in the HS games with the $n$-search among $k$ locations and coherence constraints  \mbox{$|\boldsymbol{z}_t(i) - \boldsymbol{z}_t(i+1)| \le \kappa, \forall i\in [n]$} for a fixed $\kappa \in [0,k-1]$. 

\begin{definition}[HS Graph]
The graph $G_{k,\kappa,n}$ is a DAG that~contains:
\begin{trivlist}
\item[$(i)$] \mbox{$N:=2+kn$} vertices arranged into $n+2$ layers. Layer~$0$ and Layer $(n+1)$, each contains only one vertex, respectively labeled $s$--the source vertex and $d$--the destination vertex. Each Layer \mbox{$i \in \{1,\ldots, n\}$} contains $k$ vertices whose labels are ordered from left to right by $(i,1), (i,2),\ldots, (i,k)$. % 
\item[$(ii)$] There are directed edges from vertex $s$ to every vertex in Layer $1$ and edges from every vertex in Layer $n$ to vertex $d$. For $i \in \{1,2,\ldots, n-1\}$, there exists an edge connecting vertex $(i,j_1)$ to vertex $(i+1,j_2)$ if $|j_1 - j_2| \le \kappa$.
\end{trivlist}
\end{definition}
The graph $G_{k,\kappa,n}$ has {$E\!=\!2k\!+\!(n\!-\!1)\left[k\!+\!\kappa(2k\!-\!\kappa\!-\!1)\right]\! =\! \mO(nk^2)$} edges and at least $\Omega(\kappa ^{n-1})$ paths from $s$ to $d$. The edges ending at vertex $d$ are the auxiliary edges that are added just to guarantee that all paths end at $d$; these edges do not represent any intuitive quantity related to the game. For the remaining edges, any edge that ends at the vertex $(i,j)$ represents choosing the location $j$ as the $i$-th move. In other words, a path starting from $s$, passing by vertices \mbox{$(1,j_1), (2,j_2), \ldots, (n,j_n)$} and ending at $d$ represents the $n$-search that chooses location $j_1$, then moves to location $j_2$, then moves to location $j_3$, and so~on. 
\begin{proposition}
\label{Propolayer2}
Given $k, \kappa$ and $n$, there is a one-to-one mapping between the action set $S_{k,\kappa,n}$ of the learner in the HS game (with $n$-search among $k$ locations and coherence constraints with parameter $\kappa$) and the set of all paths from vertex $s$ to vertex $d$ of the graph $G_{k,\kappa,n}$.
\end{proposition}

%%%%%%%%%%%%%%%%%%%%%%%%%%%%%%
%%%%%%%%%%%%%%%%%%%%%%%%%%%%%
%%%%%%%%%%%%%%%%%%%%%%%%%%%%%%
\section{\Algo/ Algorithm and OSMD Algorithm in the CB and HS Games}
\label{prooflogarithm}
$(i)$ As stated in Section~\ref{sec:games}, the observation graphs in the CB games are non-symmetric and they satisfy assumption $(A0)$. If we choose $\beta = \beta^*$ as in~\eqref{eq:beta^*2}, then $\beta$ satisfies~\eqref{{eq:beta2}}. Moreover, \mbox{$\beta = \mO(1/ \sqrt{TnE})$}; thus, \mbox{$M = \mO(E^2 \sqrt{TnE})$}. From~\eqref{RTCB}, the expected regret of \Algo/ in this case is bounded by \mbox{$\mO\sqrt{Tn (\alpha_{CB}) \ln{M} \ln(P)}$} (recall that \mbox{$\alpha_{CB}= kn$} is an upper bound of independence numbers of the observation graphs in the CB games). Therefore, to guarantee that this bound is better than the bound of the \textsc{OSMD} algorithm (that is $\sqrt{2TnE}$), the following inequality needs to hold:
\begin{align*}
                & \mO\left( \alpha_{CB} \cdot \ln{M} \ln(P) \right) \le E \\
\Rightarrow    & \mO\left(nk \cdot \ln{(E^2 \sqrt{TnE})} \ln(2^n) \right) \le n k^2 \\
\Rightarrow     & \mO\left( \ln{(E^2 \sqrt{TnE})} \ln(2^n) \right) \le k \\ 
\Rightarrow     & \mO\left(n \ln{(n^3 k^5 \sqrt{T}}) \right) \le k.
\end{align*}

$(ii)$  As stated in Section~\ref{sec:games}, the observation graphs in the HS games with condition $(C1)$ are symmetric and do not satisfy assumption $(A0)$. If we choose \mbox{$\beta = 1/\sqrt{n \alpha T}$} then by~\eqref{RTHS1}, we have that $R_T$ is bounded by $\mO\left(n \sqrt{\alpha_{HS} T \ln(P)} \right)$ (recall that $\alpha_{HS}= k$ is an upper bound of the independence numbers of the observation graphs in the HS games). Therefore, to guarantee that this bound is better than the bound of the \textsc{OSMD} algorithm in HS games, the following inequality needs to hold:
\begin{align*}
                & \mO\left( \alpha_{HS} \cdot n  \ln(P) \right) \le E \\
  \Rightarrow   & \mO\left( k \cdot n  \ln(P) \right) \le n k^2 \\
\Rightarrow     & \mO\left( \ln(P)  \right)\le k \\ 
\Rightarrow     & \mO\left(n \ln{\kappa}   \right) \le k.
\end{align*}

$(iii)$ Finally, the observation graphs in the HS games with condition $(C2)$ are non-symmetric and do not satisfy assumption $(A0)$. Therefore, if we choose $\beta = \beta^*$ as in~\eqref{eq:beta^*}, then $\beta$ satisfies~\eqref{eq:tune_beta1}. In this case, $\beta = \mO(1/\sqrt{TnE})$ and $M = \mO(E^2 \sqrt{TnE})$. Therefore, from~\eqref{eq:RTBOU}, in this case, $R_T$ is bounded by $\mO(n\sqrt{T \alpha_{HS} \ln{\alpha_{HS}} \ln({nM})}$. Therefore, to guarantee that this bound is better than the bound of \textsc{OSMD} (that is, $\sqrt{2TnE}$), the following inequality needs to~hold:
\begin{align*}
                & \mO\left( \alpha_{HS} \cdot n \ln{nM} \ln(P) \right) \le E \\
\Rightarrow     & \mO\left(  n k \ln{(\kappa^n) \ln(nE^2\sqrt{TnE})} \right) \le n k^2 \\ 
\Rightarrow     & \mO\left( n \ln{\kappa} \ln{(n^4 k^5 \sqrt{T}}) \right) \le k.
\end{align*}

\end{document}